\documentclass[12pt]{article}

\usepackage{float}
\usepackage{graphicx}
\graphicspath{{./}{figures/}}

\usepackage[T1]{fontenc}
\usepackage{lmodern}
\usepackage[utf8]{inputenc}

\usepackage{amsmath}
\usepackage{amssymb}
\usepackage{amsthm}
\newtheorem{theorem}{Theorem}
\newtheorem{proposition}{Proposition}
\newtheorem{lemma}{Lemma}
\newtheorem{assumption}{Assumption}
\newtheorem{corollary}{Corollary}
\newtheorem{obs}{Observation}
\theoremstyle{remark}
\newtheorem{remark}{Remark}
\theoremstyle{plain}

\usepackage{setspace}
\renewcommand{\phi}{\varphi}
\usepackage[english]{babel}
\usepackage[round,authoryear]{natbib}
\usepackage[hidelinks]{hyperref}
\renewenvironment{proof}[2]{\noindent\textbf{Proof #1 #2}\par}{\qedsymbol\normalsize\noindent\par\medskip}

\author{Constantine Sorokin\thanks{\raggedright University of Glasgow.
\href{mailto:constantine.sorokin@glasgow.ac.uk}{\mbox{constantine.sorokin@glasgow.ac.uk}}.}
\and Eyal Winter\thanks{\raggedright The Federmann Center for the Study of
Rationality, The Hebrew University of Jerusalem and Lancaster University.
\href{mailto:eyal.winter@mail.huji.ac.il}{\mbox{eyal.winter@mail.huji.ac.il}}.}
}

\title{Informational Effects in Conflict Escalation}
\date{}

\begin{document}

\maketitle

\begin{singlespace}
\begin{abstract}

We study a war of attrition in which each party chooses how far to escalate
before conceding.  If neither concedes before escalation reaches a random
disaster threshold, both receive nothing; otherwise the confrontation ends
peacefully.  When the parties are ex-ante symmetric and their benefits are
independent, improving both parties' private information about their own
stakes raises both the probability of peace and each party's ex-ante
equilibrium utility.  The result holds for every nontrivial mean-preserving
spread of posterior values---each party's expected benefit from prevailing
given its private information.  Higher stakes work in the opposite
direction: a first-order stochastic increase in posterior values induces
greater persistence, lowers the probability of peace, and can reduce both
parties' ex-ante utilities.  Starting from symmetry, a sufficiently small
improvement in only one party's information also promotes peace; its
first-order effect is exactly half that of improving both parties'
information.  At an asymmetric starting point, however, even a small
one-sided improvement can instead lower the probability of peace.


\end{abstract}
\end{singlespace}

\pagebreak

\section{Introduction}\label{sec:introduction}

International conflicts rarely move directly from disagreement to their most
destructive form.  They pass through stages at which either party can still
concede, but at which continued confrontation carries a growing risk of a
disaster event that neither party controls and that both would rather avoid.
Examples include the passage from political confrontation to violence, from
a limited engagement to a wider war, from conventional operations to nuclear
use, and from a local dispute to a broader economic disruption.  This is the
structure Schelling identified: a party does not choose disaster, but how
long to remain where disaster can happen.  Brinkmanship is a ``threat that
leaves something to chance'' \citep{Schelling,N1}, and recent work on
inadvertent and multidomain escalation gives concrete content to the stages a
confrontation may pass through \citep{Posen,Acton,SchramChance,Wan}.

Parties in such a confrontation are usually uncertain about each other.  They
may also be uncertain about themselves: about the domestic political,
strategic, or economic value that prevailing would ultimately have for them.
This paper asks what happens when that second uncertainty is reduced.  Suppose
each party comes to assess its own benefit from prevailing more accurately,
with neither becoming more optimistic on average.  Does the confrontation
become more or less likely to end before disaster?

We study a continuous escalation game between two ex-ante symmetric parties.
Each privately observes an estimate of its own benefit from prevailing, drawn
independently from a common distribution, and chooses the escalation level at
which it will concede.  The party that concedes first receives a normalized
safe payoff and its opponent receives that payoff plus its privately assessed
benefit; if the disaster event occurs before either concedes, both receive
nothing.  Neither party controls when the disaster occurs, but their
concession decisions determine how long the unresolved confrontation remains
exposed to it.  We call the probability that someone concedes before disaster
the \emph{probability of peace}---literally, the probability of resolution
before the disaster event.  Its complement measures the risk of disaster
and any resulting harm to outsiders, whereas equilibrium utility also
reflects who prevails and how much prevailing is worth.  We therefore rank both objects, and show
that they move together in the symmetric information comparison.

Our main result is that better self-assessment promotes peace.  A more
accurate signal about a party's own benefit spreads the distribution of its
posterior expected benefit in convex order while leaving its mean unchanged
\citep{Blackwell,MM,GanuzaPenalva}.  When both parties' signals improve in this
sense and their benefits are independent, every such mean-preserving spread
raises the probability of peace and each party's ex-ante equilibrium utility.
The comparison is global for a broad class of bounded-support distributions
(Assumption \ref{ass:regularity}): it is not tied to a particular signal
technology or a small perturbation.  Information acquisition is exogenous
throughout: we compare information structures rather than deriving them from
an acquisition or persuasion problem.

The mechanism is separation: the earlier concessions induced by lower
posterior values outweigh the additional persistence induced by higher ones,
because the confrontation ends as soon as either party concedes.  Welfare
also benefits from sorting: conditional on peace, the higher-value party
prevails.  The welfare proof shows that this allocation gain survives the
endogenous change in the probability of peace.

This reasoning is suggestive but insufficient.  A mean-preserving spread
moves probability toward both tails, and larger benefits by themselves make
parties more persistent, so the informal story supplies forces in both
directions without determining their net effect.  The equilibrium concession
strategy is also not fixed while the distribution moves: it is determined by
the whole distribution.  The general distributional comparative-statics results in \citet{Distr}
do not determine which force dominates here.  Our proof therefore works
directly with quantile functions and recomputes equilibrium along the
information comparison.  It
also yields a broader ordering of planned exposure at the first concession.  The welfare
result requires a further allocation argument and does not follow from the
peace result alone.

To make the comparison unambiguous, we must also specify which equilibrium
is used.  Incomplete-information wars of attrition can admit asymmetric
equilibria even when the parties are ex-ante identical.  We therefore use
the standard common wild-type refinement: each party has a small probability
of never conceding, and we take the limit as that probability vanishes.
For every positive probability of a wild type the resulting equilibrium is
unique; symmetry of the
environment then makes it symmetric, so the symmetric equilibrium used in the
main comparison is selected rather than assumed.  The same refinement also
selects the asymmetric equilibria studied in the one-sided comparison.

Measuring escalation by accumulated disaster hazard makes the selected
concession strategy coincide with the bid function of the standard
private-value war of attrition with linear costs.  It also makes the
probability of peace independent of the disaster-threshold distribution
under Assumption \ref{ass:regularity}, paralleling the neutrality result
for equilibrium damage in \citet{N1}.

Increases in information accuracy should be distinguished from increases in
the stakes, which work in the opposite direction.  A first-order stochastic
increase in benefits lowers the probability of peace: parties with more to
gain are more persistent.  Their ex-ante utilities need not rise either,
because the larger value of prevailing may be outweighed by greater exposure
to disaster.  The contrast between the two comparisons is the paper's central
point.  Raising what the parties stand to gain intensifies escalation;
raising how well they understand it, while holding the mean fixed, promotes
resolution.

Does the favorable information effect require both parties to learn more?
Starting from symmetry, a sufficiently small improvement in only one party's
information also raises the probability of peace.  This comparison uses the
same equilibrium selection, allowing the lowest types to concede immediately
when the resulting asymmetry makes that optimal.  The one-sided theorem
states the conditions on the initial distribution and the fixed direction
of change under which the local result holds.  At symmetry, its first-order effect is exactly
half the effect of improving both parties' information.  The result is local
and directional: the admissible size of the change may depend on its
direction, and we make no corresponding one-sided welfare claim.
The conclusion is genuinely local: at an asymmetric starting point, even
a small one-sided improvement can instead lower the probability of peace
(Theorem \ref{onesided}(ii)).

The symmetric result has implications beyond the two parties: if disaster
also imposes a loss on outsiders, better self-assessment reduces that risk.
Thus the parties' ex-ante welfare and the external-risk effects point in the
same direction.

Although international conflict is the leading application, the same
structure arises when political leadership contests damage the organization
being contested, price
wars threaten the survival of firms, patent races remain exposed to entry, or
animal contests risk injury or predation.  The last application connects to
the biological origins of war-of-attrition models and the role of information
about asymmetries in animal conflict \citep{Smith}.

\subsection*{Related literature}

The paper builds first on theories of brinkmanship and staged escalation.
\citet{Schelling} introduced threats that leave something to chance,
and \citet{N1} showed that equilibrium damage in a continuous
brinkmanship game is neutral to the technology mapping escalation into risk.
\citet{Powell1988} studies nuclear brinkmanship with two-sided
incomplete information.  \citet{Posen}, \citet{Acton},
\citet{SchramChance}, and \citet{Wan} describe mechanisms through which a
limited confrontation may cross into a more costly stage.  We use a
stochastic disaster event to represent this transition without committing
to a particular escalation ladder.  Our question is how information about
the parties' own benefits changes the probability that concession occurs
before the transition.

Formally, the game is closely related to wars of attrition with incomplete
information \citep{FudTir,BulowKlemperer,Clara,w-o-a-all-pay}.
\citet{w-o-a-all-pay} develop the connection between the war of attrition and
the all-pay family, while \citet{BulowKlemperer} use revenue equivalence to
analyze the generalized war of attrition.  \citet{Myatt} studies asymmetric
distributions of
privately known values and the resulting immediate-exit behavior.
Equilibrium multiplicity is emphasized by \citet{CastilloMiranda}; we use
the common wild-type refinement to select an equilibrium throughout.
\citet{Gieczewski} and \citet{Georgiadis} instead introduce stochastically
evolving public states.  Equilibrium characterization becomes substantially
more delicate with asymmetric, correlated, or interdependent values
\citep{Siegel2014,RentschlerTurocy,LuParreiras,ChiMurtoValimaki}.  We retain
independent private values in order to isolate the comparative statics of the
distribution of posterior benefits; the paper makes no claim for
interdependent-value environments.

The informational question is already explicit in \citet{N1}, who asks how
improved information changes brinkmanship.  The closest endpoint comparison
is \citet{MorathMuenster2008}, who compare private values with public
revelation of all values across auction formats.  Complete information
preserves bidders' interim payoffs and lowers expected revenue in the
standard all-pay auction; in mixtures with a positive war-of-attrition
component, it raises bidders' payoffs and lowers revenue under their
equilibrium selection.  The payoff direction is consistent with Theorem
\ref{welfare}, but their information becomes public, their comparison is
between two endpoints, and the war-of-attrition limit requires an equilibrium
selection.  Here information remains private, only self-assessment becomes
more accurate, and Theorem \ref{mps} covers every symmetric convex-order
increase.

The information-acquisition papers ask a related but different question.
\citet{MorathIgnorance} study observable acquisition of one's own binary cost
before finite-horizon attrition, \citet{IA_woa} study acquisition and
verifiability of information about the state at a stochastic deadline, and
\citet{MorathMuenster2013} study learning one's own value in a first-price
all-pay contest.  More recent work instead studies acquisition of information
about an opponent's value \citep{Chen2025,FengSong2025}.  These papers
characterize acquisition, observability, effort, and payoff effects for
particular signal technologies.

Disclosure and other specific information regimes are examined by
\citet{Warneryd}, \citet{Seel}, \citet{KotowskiLi},
\citet{KovenockMorathMuenster}, and \citet{EwerhartLareida}.  Our comparison
instead ranks arbitrary exogenous, symmetric improvements in private
information about one's own benefit.

The information-ordering literature supplies the interpretation of that
comparison.  \citet{Blackwell} and \citet{MM} provide the classical
foundations, and we follow \citet{SOrders} for stochastic-order terminology.
For scalar posterior expectations, the convex order used here is called
\emph{integral precision} by \citet{GanuzaPenalva}.  Our contribution is not
another characterization of informativeness, but a global equilibrium
comparative static under this standard order.  \citet{Distr} develops general distributional
comparative statics that allow equilibrium strategies to change with the
distribution.  The sufficient monotonicity conditions used there do not
directly rank the nonlinear survival statistic here, which is why we use a problem-specific quantile argument.

Bayesian persuasion and information design instead choose an information
structure to optimize a sender's objective \citep{BPers,BergemannMorris}.
\citet{S_W_N} do so for auction
revenue, an objective linked by revenue equivalence to expected attrition
dissipation, while \citet{Szydlowski} designs information about
relative strength for a sender who wishes to sustain conflict.  Theorem
\ref{mps} orders the nonlinear probability of peace generated by an exogenous
improvement in the parties' private information about their own stakes.

A neighboring literature studies private information, mediation, and outside
resolution in conflict and bargaining.  \citet{Fearon1995} emphasizes
private information and incentives to misrepresent as a rationalist source of
war.  \citet{Horner}, \citet{Zheng}, and
\citet{Hennigs} study mediation or information provision aimed at
preventing conflict; \citet{Basak} studies the social value of public
information in bargaining.  \citet{MaManove}, \citet{FuchsSkrzypacz},
\citet{FanningDeadline}, \citet{EkmekciZhang}, and \citet{Herrera} analyze
bargaining with deadlines, imperfect control, or external resolution.
Our game has no offers, division of surplus, or strategic use of an outside
forum: at every escalation level a party can only concede or continue.

The rest of the paper is organized as follows.  Section~\ref{sec:model}
presents the model, characterizes equilibrium, and derives the probability of
peace.  Section~\ref{sec:stakes} compares changes in benefits and in
information and establishes the welfare result.  Section~\ref{sec:onesided}
studies one-sided information improvements, and Section~\ref{sec:discussion}
discusses implications and limitations.

\section{Model and equilibrium}\label{sec:model}
\subsection{Posterior values, escalation, and payoffs}
There are two ex-ante identical parties, indexed by $i=1,2$.  Party $i$'s
actual benefit from prevailing is $\widetilde\theta_i$.  Before escalation
begins, the party privately observes a signal $s_i$ and forms the posterior
value
\[
\theta_i=E[\widetilde\theta_i\mid s_i].
\]
The pairs $(\widetilde\theta_i,s_i)$ are independent across parties.  Hence
the posterior values $\theta_1$ and $\theta_2$ are independent draws from a
common distribution $F$.  We refer to $\theta_i$ as party $i$'s posterior
value, or simply as its type.

The confrontation moves along a continuous escalation scale.  A random
disaster threshold $T$, with distribution $G$, determines the escalation
level at which disaster occurs.  We call $G$ the \emph{escalation technology}:
$G(t)$ is the probability that disaster has occurred by level $t$.

The confrontation begins at level $0$ and moves up the scale until a party
concedes or level $T$ is reached, whichever comes first.  Neither party
observes $T$ before it is reached.  The distributions $F$ and $G$ and the
payoff structure are common knowledge, and the parties are risk neutral.
Because a concession ends the interaction, only the lower of the two chosen
concession levels matters.  We index escalation by a scalar level and
describe movement up the scale in temporal language such as ``sooner'' and
``longer''; nothing depends on the scale being calendar time.

No new private signals arrive during escalation.  While the confrontation
continues, each party observes the current level and the fact that neither
concession nor disaster has occurred.  Continued non-concession can change
beliefs about the opponent, but there is only one continuing public history
at each level.  A pure stopping rule can therefore be represented by a
concession level $a_i\geq0$; this representation does not require commitment
to ignore later information.  If party $i$ concedes
first and does so before $T$ is reached, it receives the normalized safe
payoff $1$, while party $j\neq i$ receives $1+\widetilde\theta_j$.  If
$T<\min\{a_1,a_2\}$, disaster occurs first and both receive zero.  Ties are
settled by a fair coin, but they occur with probability zero in the
equilibrium studied below.  We call the party that does not concede first the \emph{winner}, conditional on
peace.

Because the parties are risk neutral and
$E[\widetilde\theta_i\mid s_i]=\theta_i$, expected payoffs depend on a
party's signal only through its posterior value $\theta_i$.  Conditional on
the signal, the three payoff levels are thus $0$ under disaster, $1$ on
conceding, and $1+\theta_i$ on prevailing; only $\theta_i$ varies across
types, as in \citet{N1}.  The interaction
is therefore a private-value game in $\theta$, and an improvement in
information changes the distribution $F$ of posterior values.

To characterize behavior in this game, we impose the following regularity
and independence conditions.

\begin{assumption}[Maintained conditions]\label{ass:regularity}
\emph{(a) Posterior values.}  The distribution $F$ has bounded support
$[\underline\theta,\bar\theta]$, with
$0\leq\underline\theta<\bar\theta<\infty$, and a density $f$ that is
continuous and positive on the interior of its support.

\emph{(b) Escalation technology.}  The distribution $G$ has support $[0,\infty)$, satisfies $G(0)=0$ and
$\lim_{t\to\infty}G(t)=1$, and has a continuous positive density $g$.  The
threshold $T$ is independent of
$(\widetilde\theta_1,s_1,\widetilde\theta_2,s_2)$.
\end{assumption}

The equilibrium comparison balances the prospect of an opponent conceding
against the risk of disaster.  We therefore use survival probabilities and
hazard rates for both distributions.  Write $\mu=E[\theta_i]$ for the
common mean posterior value and
\[
S_F(x)=1-F(x),\qquad S_G(t)=1-G(t)
\]
for the two survival functions, and
\[
\lambda_F(x)=\frac{f(x)}{S_F(x)},\qquad
\lambda_G(t)=\frac{g(t)}{S_G(t)}
\]
for their hazard rates.

These probabilities also provide a common scale for comparing different
escalation technologies.  Rather than use physical levels, measure
escalation by accumulated disaster hazard, which we call the
\emph{disaster-hazard clock}: let
\begin{equation}\label{hazard-clock}
H(t)=-\log S_G(t).
\end{equation}
Thus $H(t)$ is the hazard accumulated by escalation level $t$, and surviving
to that level has probability $e^{-H(t)}$.  Assumption
\ref{ass:regularity} makes $H$ a strictly increasing map from physical
escalation levels onto $[0,\infty)$.

\subsection{The selected equilibrium}

Before comparing outcomes across information structures, we need an
equilibrium selection that applies consistently to all of them.  Let each
party be an \emph{ordinary type}, behaving as described above, with probability
$c\in(0,1)$ and a \emph{wild type} that never concedes with probability
$1-c$.  We call this the common wild-type perturbation.  The
wild-type-perturbed game has a unique equilibrium; when the
ordinary types of both parties have the same distribution $F$, symmetry of
the perturbation and uniqueness make that equilibrium symmetric and
increasing.  As $c\uparrow1$, the wild-type probability vanishes; this limit
selects the symmetric equilibrium used below.  Write
$\tau(\theta)$ for the common limiting physical concession strategy, so party $i$
chooses $a_i=\tau(\theta_i)$.

The refinement supplies the boundary condition that is absent from the
unperturbed war of attrition.  Without wild types, the highest ordinary type
escalates without bound, so the point at which the parties' concession
strategies meet at the top is not pinned down; the resulting free integration
constant permits asymmetric equilibria even in a symmetric environment.  A
positive probability of a never-conceding opponent instead gives ordinary
types a finite upper concession level.  No ordinary type wants to continue
past the opponent's highest ordinary concession, because beyond it the
opponent must be wild.  The two finite concession ranges consequently have a
common upper endpoint.  We call this the \emph{common-upper-endpoint
condition}; it pins down the integration constant and the equilibrium.

This is the standard wild-type refinement of \citet{NalebuffRiley}.  On the
disaster-hazard clock \eqref{hazard-clock}, survival to accumulated hazard
$H$ has probability $e^{-H}$.  It therefore acts as an exponential discount
factor, giving the attrition structure characterized by \citet{Clara}.  The Appendix proves the
boundary argument and takes the limit, while the Supplementary Appendix gives
the full derivation.  Unless stated otherwise, all equilibrium objects refer
to this selected limit.

To characterize the selected strategy, we examine deviations to other
concession levels.  If a party with posterior value
$\theta$ chooses the level assigned to type $x$, its expected payoff is
\begin{equation}\label{EU}
U(x,\theta;\tau)
=S_G(\tau(x))S_F(x)
+(1+\theta)\int_{\underline\theta}^x
S_G(\tau(y))\,dF(y).
\end{equation}
The first term is the payoff from conceding before both the opponent and
disaster.  The second is the payoff from outlasting an opponent of type
$y<x$, with the survival probability evaluated at that opponent's
concession level.

We call the accumulated disaster hazard a type plans to bear before
conceding its \emph{planned exposure}.  The proposition characterizes the
strategy both in physical levels and in these hazard units.

\begin{proposition}[Selected symmetric equilibrium]\label{FOC}
In the equilibrium selected by the common wild-type refinement, the limiting
common physical concession strategy is characterized by
$\tau(\underline\theta)=0$ and
\begin{equation}\label{FOC-i}
\lambda_G(\tau(x))\tau'(x)=x\lambda_F(x).
\end{equation}
Equivalently, on the disaster-hazard clock, the same strategy is represented
by the planned-exposure function
\begin{equation}\label{beta-definition}
\beta_F(x):=H(\tau(x))
=\int_{\underline\theta}^x z\lambda_F(z)\,dz,
\end{equation}
and hence
\begin{equation}\label{equilibrium-survival}
S_G(\tau(x))=\exp[-\beta_F(x)].
\end{equation}
The physical concession strategy $\tau$ is strictly increasing and is the
unique physical concession strategy selected by the refinement.
\end{proposition}

The condition balances two marginal hazards.  Waiting a little longer gives
type $x$ a chance to obtain the extra benefit $x$ when the opponent concedes;
the rate of that event is $\lambda_F(x)$.  The same delay increases exposure
to disaster at rate $\lambda_G(\tau(x))\tau'(x)$.  Equation
\eqref{FOC-i} equates the marginal gain and loss.

The Appendix derives the wild-type-perturbed equilibrium, takes the wild-type limit,
and shows that this condition is sufficient, not merely necessary.  Once
\eqref{FOC-i} holds, the derivative of
\eqref{EU} with respect to the imitated type $x$ has the sign of
$\theta-x$.  Choosing the concession level assigned to one's own type is
therefore the global best response, and the second-order condition is strict
at every interior type.

The planned exposure $\beta_F(x)$ gives the probability that disaster has
not occurred by type $x$'s concession level as $e^{-\beta_F(x)}$.  The escalation technology $G$ determines only how this planned exposure is converted back into a physical
concession level through $\tau(x)=H^{-1}(\beta_F(x))$.  Thus working with
$\beta_F$ separates planned exposure from the particular escalation
technology.  The same function also coincides with the equilibrium bid
function of the standard symmetric private-value war of attrition with
linear costs.

\subsection{The probability of peace}

The strategy tells us how far each type plans to escalate.  To obtain the
probability of peace, we must average the probability of avoiding disaster
until the first planned concession.  Quantile ranks make this averaging
especially simple, because they are uniform whatever the distribution of
posterior values.  With $\nu=F^{-1}$, write
\[
\mathcal Z_\nu(u)=\beta_F(\nu(u))
=\int_0^u\frac{\nu(v)}{1-v}\,dv,
\qquad
\sigma_\nu(u)=e^{-\mathcal Z_\nu(u)}.
\]
The substitution $z=\nu(v)$ in \eqref{beta-definition} gives the integral,
since $\lambda_F(\nu(v))\nu'(v)=1/(1-v)$.  Thus $\mathcal Z_\nu(u)$ is
planned exposure at rank $u$, and $\sigma_\nu(u)$ its disaster-survival
probability.  Write $q[F]=q[\nu]$ for the selected
probability of peace, using the distribution argument for comparisons and
the quantile argument for the calculation below.

The lower-value party concedes first, so the relevant rank is the minimum
of two independent uniform ranks, which has density $2(1-u)$.  Averaging
its survival probability gives
\begin{equation}\label{peace1}
\begin{aligned}
q[\nu]&=2\int_0^1(1-u)\sigma_\nu(u)\,du\\
&=2\int_0^1(1-u)
\exp\left[-\int_0^u\frac{\nu(v)}{1-v}\,dv\right]du.
\end{aligned}
\end{equation}

The Appendix also derives this expression from type coordinates.  Notice that
$G$ has disappeared from the expression, yielding the following neutrality
result.

\begin{obs}[Neutrality of the escalation technology]\label{indie}
For a fixed distribution of posterior values, the selected equilibrium
probability of peace is independent of the escalation technology $G$.
\end{obs}

The absence of $G$ from the functional reflects exact behavioral adjustment.
A more imminent disaster induces every type to choose a lower physical
concession level, while a less imminent disaster induces a higher one.  In
either case, the accumulated hazard borne before concession is unchanged.
This neutrality result is the counterpart of \citet{N1}.  The remaining comparisons therefore concern the distribution of posterior
values: we distinguish increases in their level from improvements in the
information on which they are based.

Alongside $q[F]$, we compare $\overline U[F]$, either party's ex-ante
equilibrium utility when both posterior values have distribution $F$.
A distributional comparison is \emph{nontrivial} if the two quantile
functions differ on a set of positive Lebesgue measure.

\section{Stakes and information}\label{sec:stakes}

\subsection{Higher stakes}

We first change the level of the parties' posterior values rather than the
accuracy of their information.  Say that $F^1$ first-order stochastically
dominates $F^0$ if its quantile function is weakly higher at every rank.

\begin{proposition}[Higher stakes lower peace]\label{FOD}
If $F^1$ first-order stochastically dominates $F^0$, then
$q[F^1]\leq q[F^0]$, strictly under a nontrivial dominance comparison.
\end{proposition}

The intuition is direct from \eqref{peace1}.  At every quantile rank, a
higher value makes the lower-value party willing to bear more accumulated
disaster hazard before conceding.  Greater stakes therefore reduce the
probability of peace.

Higher stakes do not necessarily make the parties better off.  They raise
the benefit obtained by the winner, but they also intensify escalation.
The second force can dominate even in a familiar one-parameter comparison.

\begin{proposition}[Higher stakes can lower equilibrium utility]
\label{fod-welfare}
A first-order stochastic increase in posterior values can strictly lower
both parties' ex-ante equilibrium utility.
\end{proposition}

A simple example establishes the claim; the calculations are in the
Appendix proof of Proposition \ref{fod-welfare}.  On $[0,1]$, let $F^0(x)=x$ and
$F^1(x)=x^2$.  The latter distribution
first-order stochastically dominates the former: low values become less
likely and the mean rises from $1/2$ to $2/3$.  Nevertheless, the probability
of peace falls from $4e-10\approx0.8731$ to
$(9-e^2)/2\approx0.8055$, and each party's ex-ante equilibrium utility falls
from $3e-7\approx1.1548$ to approximately $1.1198$.

\citet[Proposition~4]{N1} gives a useful complementary comparison.  In his
notation, a first-order shift of the
cost-benefit index $X$ toward lower costs raises both parties' expected
utilities.  Under our payoff normalization,
\[
X=\frac{\theta}{1+\theta}.
\]
Here $X$ is the loss from conceding as a fraction of the winner's payoff, so
$1-X=1/(1+\theta)$ is the conceder's payoff relative to the winner's.
Thus \citet{N1} moves the relative stakes downward, whereas our first-order
comparison moves both $\theta$ and $X$ upward.  There is no contradiction:
making concession relatively less costly can benefit both parties, while
raising the benefit from prevailing can intensify escalation enough to make
both worse off.

We next hold the mean value fixed and compare equilibrium outcomes across
information structures.

\subsection{More accurate information}\label{sec:information}

Return to the actual benefit $\widetilde\theta_i$.  A more informative signal
about that benefit makes the posterior expectation $\theta_i$ more dispersed
without changing its ex-ante mean.  For $k=0,1$, let $\theta^k$ have
distribution $F^k$.  We say that $F^1$ is a \emph{mean-preserving spread}
of $F^0$, written $\theta^0\leq_{cx}\theta^1$, if
\[
E[\varphi(\theta^1)]\geq E[\varphi(\theta^0)]
\]
for every convex function $\varphi$ for which the expectations exist.
\citet{Blackwell} and \citet{MM} provide the classical foundations for
comparing experiments and their economic value.  For scalar posterior
expectations, our mean-preserving-spread comparison is the \emph{integral
precision} criterion of \citet{GanuzaPenalva}.  A Blackwell improvement means
that the old signal can be obtained by garbling the new one; it implies
this convex-order relation between posterior values.  For a fixed prior,
however, not every mean-preserving spread is feasible as a posterior-value
distribution.  Our results cover all such spreads; Section
\ref{sec:bounds} returns to feasibility for a given prior.

To determine what this change in information does to peace, we follow the
exposure generated by the first planned concession.  Let
$\theta_{(1)}=\min\{\theta_1,\theta_2\}$ and
\[
Z_F=\beta_F(\theta_{(1)})
\]
denote the planned exposure at the lower-value party's concession level.  It
is the accumulated disaster hazard at the first planned concession level---
the hazard the confrontation would run if disaster did not intervene.
Equivalently, $Z_F=\mathcal Z_\nu(F(\theta_{(1)}))$: it is the
planned-exposure function evaluated at the random rank of the lower type.
The survival representation is
\[
q[F]=E[e^{-Z_F}].
\]

\begin{theorem}[Information accuracy and peace]\label{mps}
Suppose that $F^1$ is a mean-preserving spread of $F^0$ and both
distributions satisfy Assumption \ref{ass:regularity}(a).  Then
\[
q[F^1]\geq q[F^0].
\]
The inequality is strict for a nontrivial convex-order comparison.
\end{theorem}

More accurate information separates posterior values.  A party that learns
that prevailing is worth little concedes sooner; a party that learns that it
is worth a great deal persists longer.  Because the confrontation ends when
either party concedes, the behavior of the lower posterior value determines
planned exposure at the first concession.

This interpretation does not require the information improvement to place
more probability near the bottom of the overall distribution.  A
mean-preserving spread may be confined to an arbitrarily high range of
posterior values, leaving every lower type unchanged.  Even then, some types
within the affected range move downward and concede earlier, while others
move upward and persist longer.  Separation therefore need not involve
the bottom of the overall distribution.

This separation intuition is incomplete as a proof.  The probability of
peace is a nonlinear survival functional evaluated at the equilibrium
concession level of the lower-value party.  A mean-preserving spread changes
both the distribution of the minimum $\theta_{(1)}$ and the planned-exposure
function $\beta_F$.  Thus the earlier concessions
and additional persistence cannot be compared while holding behavior fixed,
and no direct Jensen argument yields the result.  The theorem establishes that
the first effect dominates after the entire equilibrium response is taken
into account.

The proof handles this adjustment by connecting the two quantile functions
with a linear path and recomputing equilibrium at every point.  The
probability of peace increases along the entire path, giving a global
rather than merely local comparison.  The same argument works for more
than the survival function $e^{-z}$: it yields the following broader
ordering of planned exposure.

\begin{proposition}[Planned-exposure ordering]\label{mps-dcx}
Suppose that $F^1$ is a mean-preserving spread of $F^0$ and both
distributions satisfy Assumption \ref{ass:regularity}(a).  Then, for every
continuously differentiable decreasing convex function $\Phi$,
\begin{equation}\label{dcx-ordering}
E[\Phi(Z_{F^1})]\geq E[\Phi(Z_{F^0})].
\end{equation}
The inequality is strict if the convex-order comparison is nontrivial and
$\Phi'<0$ everywhere.
\end{proposition}

The role of Proposition \ref{mps-dcx} is to strengthen and interpret the
peace result.  The choice $\Phi(z)=e^{-z}$ gives Theorem \ref{mps}; the
choice $\Phi(z)=-z$ gives the more direct conclusion that better information
lowers expected planned exposure at the first concession.  Thus the increase in
the probability of peace is not an isolated property of the exponential
survival formula.  It is one consequence of a broader improvement in the
equilibrium distribution of exposure.  To see what changes in the confrontation, compare this exposure at the
first concession with the exposure each party individually plans to bear.

A party's \emph{individual expected planned exposure} is unchanged by better
information.  By Fubini's theorem,
\[
E[\beta_F(\theta)]=E[\theta]=\mu.
\]
What falls is \emph{expected planned exposure at the first concession},
$E[Z_F]$.  In the corresponding linear-cost war of attrition, each party
pays this elapsed hazard as its cost, so total dissipation is $2Z_F$.
Fubini's theorem gives the standard expected-dissipation identity
\[
2E[Z_F]=E[\theta_{(1)}],
\]
as derived in the Appendix proof of Proposition \ref{mps-dcx}.  This is a
benchmark from the corresponding attrition game, not the literal payoff
loss in our disaster model.
In the general case, the probability of peace depends on the full distribution of $Z_F$, which is
why its nonlinear ordering does not follow from the classical
expected-dissipation calculation
\citep{w-o-a-all-pay,BulowKlemperer,S_W_N}.

This distinction between expected planned exposure at the first concession
and its distribution also
clarifies the connection to information design.  The result of
\citet{S_W_N}, translated by revenue
equivalence to the classical two-party war of attrition, also makes full
revelation of the parties' own values optimal
\citep{w-o-a-all-pay,BulowKlemperer}.  The policy is the same, but the
objective is different: their result maximizes bidders' expected surplus,
whereas Theorem \ref{mps} maximizes the probability of peace.  Neither
result implies the other.

Theorem \ref{mps} ranks peace, not the parties' payoffs.  Information also
changes which type prevails and the payoff conditional on resolution.  The
next subsection therefore treats equilibrium welfare separately.

\subsection{Equilibrium welfare}

\begin{theorem}[Information accuracy and welfare]\label{welfare}
Under the assumptions of Theorem \ref{mps},
\[
\overline U[F^1]\geq\overline U[F^0],
\]
strictly for a nontrivial convex-order comparison.
\end{theorem}

The welfare result combines two effects.  First, more accurate information
raises the probability of avoiding disaster.  Second, whenever resolution is
peaceful, the increasing equilibrium makes the party with the higher
posterior value the winner.  Greater dispersion improves this allocation term
even after its interaction with the endogenous survival probability is taken
into account.  The latter step requires a separate argument in the Appendix;
it does not follow from Theorem \ref{mps} alone.

\subsection{Bounds and magnitude}\label{sec:bounds}

How large can the effect on peace be?  We first bound it using only the
support and mean of posterior values, then illustrate it for a fixed
uniform prior.  These are different comparisons: a fixed prior restricts
which posterior distributions can be generated by information.  The
convex-order result gives the following bounds when only support and mean
are specified.

\begin{corollary}[Bounds on the probability of peace]\label{maxminsp}
Suppose posterior values lie in
$[\underline\theta,\bar\theta]$ and have mean $\mu$, and write
\[
\pi=\frac{\mu-\underline\theta}
{\bar\theta-\underline\theta}.
\]
Across all such distributions, the probability of peace satisfies
\[
\frac{2}{2+\mu}
\leq q[F]\leq
\bigl(1-\pi^{2+\underline\theta}\bigr)
\frac{2}{2+\underline\theta}
+\pi^{2+\underline\theta}\frac{2}{2+\bar\theta}.
\]
For distributions with atoms, $q$ denotes the continuous quantile extension
of \eqref{peace1}.
\end{corollary}

If every posterior value equals $\theta_*$, the probability of peace is
$2/(2+\theta_*)$.  The lower bound evaluates this formula at the mean
$\mu$; the ceiling in \eqref{lower-support-bound} below evaluates it at
the lowest possible value.  The lower bound therefore occurs when both
parties use the constant posterior value $\mu$.
The upper bound occurs when posterior values put probability $1-\pi$ on
$\underline\theta$ and probability $\pi$ on $\bar\theta$.
These extreme distributions are not smooth, but the quantile proof of
Theorem \ref{mps} applies directly to bounded nondecreasing quantile
functions and therefore covers them.  Setting $\underline\theta=0$ recovers
the simpler upper bound $1-\mu^2/[\bar\theta(2+\bar\theta)]$.

How close can peace come to certainty if high values are allowed to become
arbitrarily large?  The lowest possible value supplies a ceiling that does
not depend on the mean or upper endpoint.  Since $\nu(u)\geq\underline\theta$, equation \eqref{peace1} implies
\begin{equation}\label{lower-support-bound}
q[F]
\leq2\int_0^1(1-u)^{1+\underline\theta}\,du
=\frac{2}{2+\underline\theta}.
\end{equation}
The right-hand side is the probability of peace if every party's value equals
the lowest feasible value.  Any probability placed above that value increases
planned exposure at the affected ranks and cannot improve on the
all-lowest-value case.

The fixed-mean limiting case now has a simple interpretation.  Hold
$\underline\theta$ and $\mu$ fixed and allow the upper endpoint to grow.  Then
\[
\lim_{\bar\theta\to\infty}
\sup_{\substack{\operatorname{supp}(F)\subseteq
 [\underline\theta,\bar\theta]\\ E_F[\theta]=\mu}}
q[F]
=\frac{2}{2+\underline\theta}.
\]
If $\underline\theta=0$, the supremum therefore converges to one.  It is
approached by placing almost all probability on the lowest value and a
vanishing probability on an arbitrarily high value so as to preserve the
mean.  Unless $\mu=\underline\theta$, the limiting value is a supremum rather
than a maximum.  In the zero-lower-bound case, the intuition is immediate:
with probability approaching one, at least one party has zero value and
concedes at once.

The bounds above allow any distribution with the specified support and
mean.  To illustrate a feasible information comparison for a particular
prior, now suppose the actual benefit is uniform on
$[0,B]$.  With no information, each party's posterior value is the constant
prior mean $\mu=B/2$; with full information, its posterior value is uniform
on $[0,B]$.  Table~\ref{tab:peace} compares the resulting probabilities of
peace.

\begin{table}[H]
\centering
\begin{tabular}{ccccc}
$B$ & $\mu$ & no information & full information & relative increase \\
\hline
1 & 0.5 & 0.800 & 0.873 & 9\% \\
2 & 1 & 0.667 & 0.792 & 19\% \\
10 & 5 & 0.286 & 0.533 & 87\%\\
\hline
\end{tabular}
\caption{\label{tab:peace}The effect of information on the probability of peace.}
\end{table}

Full revelation under a uniform prior does not attain the upper bound in
Corollary \ref{maxminsp}: the probability is $0.873$ rather than $0.917$ at
$B=1$, and $0.533$ rather than $0.792$ at $B=10$.  The corollary ranges over
all distributions with the stated support and mean, including distributions
that no signal about a uniform benefit can generate.  Given a uniform prior,
full revelation is already the convex-order maximum among feasible posterior
distributions, so the full-information values in Table~\ref{tab:peace} are
the best attainable for that prior.

\section{One-sided information}\label{sec:onesided}

So far, we have improved both parties' information together.
What if only one party learns more?  Starting from symmetry, does a small
one-sided improvement still promote peace?  And does that favorable effect
persist as the parties' information becomes increasingly different?
We first explain the asymmetric equilibrium and establish the local
effect, including why it is half the symmetric effect.  We then discuss
its implications and use an example to show how the effect can reverse.

Keeping party 2's information fixed does not keep its behavior fixed:
both concession strategies adjust when party 1's information changes.
At symmetry, equal posterior values choose equal concession levels, so
the lower-value party concedes first.  With different distributions, this
ordering need not hold.  To describe the asymmetric equilibrium, we pair
the ranks in the two distributions that choose the same positive concession
level; we call this \emph{rank matching}.  We use the same common wild-type
refinement as before, which anchors this matching at the highest types.
At the bottom, some types of one party may instead concede immediately.  Thus a one-sided comparison
must account for changes in both strategies, their matching, and possible
immediate concession.

Can these adjustments overturn the favorable information effect even when
the change is small?  To answer this, we change the \emph{distribution}
of party 1's posterior value by a small amount, holding its mean fixed.
This is a change in the information structure, not a particular favorable
or unfavorable signal realization.  For this local result, we restrict
posterior values to the fixed support $[0,1]$; this is a restriction on the
environment.  Starting from
a density $f$ on this support, we consider $f_\eta=f+\eta\delta$: the fixed direction $\delta$ describes
how probability is redistributed, and $\eta$ controls the size of the change.
We choose $\delta$ so that positive $\eta$ gives a mean-preserving spread.
Smallness alone, however, does not ensure that a density remains
nonnegative near an endpoint where it vanishes.  The theorem therefore
bounds the change relative to the initial density, keeping small changes
feasible for either sign of $\eta$.  A second bound controls how quickly
$f$ can vanish near zero, where a change in information may cause some
types to concede immediately.  The upper-endpoint density may vanish,
and both tails may change.

To state the comparison, write $q[F_1,F_2]$ for the selected probability
of peace when the parties have distributions $F_1$ and $F_2$, retaining
$q[F]=q[F,F]$ in the symmetric case.  Square brackets specify distributions
(or their quantiles); parentheses and subscripts denote scalar paths.

\begin{theorem}[One-sided information: local robustness and reversal away from symmetry]
\label{onesided}
\emph{(i) Local robustness.}  Suppose the parties initially share a
posterior-value distribution $F$ on $[0,1]$, with density $f$ continuous
on $[0,1]$ and positive on $(0,1)$.  Let $\delta$ be continuous on
$[0,1]$ and suppose that, for some constants $0\leq K<\infty$,
$\alpha\geq0$, $0<c_1\leq c_2$, and $\bar x\in(0,1)$,
\[
\begin{aligned}
|\delta(x)|&\leq Kf(x) &&(0\leq x\leq1),\\
c_1x^\alpha&\leq f(x)\leq c_2x^\alpha &&(0<x\leq\bar x).
\end{aligned}
\]
Keep party 2's distribution fixed at $F$ and give party 1 the distribution
$F_\eta$ with density $f_\eta=f+\eta\delta$.  If $F_\eta$ is a nontrivial
mean-preserving spread of $F$ for every sufficiently small $\eta>0$, and
$q(\eta):=q[F_\eta,F]$, then $q$ is differentiable at zero with
$q'(0)>0$.  In particular, there is $\bar\eta>0$ such that
\[
q(\eta)>q(0)
\qquad\text{for every }\eta\in(0,\bar\eta).
\]

\emph{(ii) Reversal away from symmetry.}  There exist asymmetric starting
points at which arbitrarily small one-sided mean-preserving spreads strictly
lower the selected equilibrium probability of peace.  These spreads can
be induced by more informative private signals in Blackwell's sense.  Thus
the favorable local effect in part~(i) need not persist away from symmetry.
\end{theorem}

A one-sided information improvement breaks symmetry, so the symmetric
result cannot be applied directly.  Both parties adjust their concession
strategies, and some of the lowest types may begin conceding immediately.
The theorem's bounds allow the proof to account for these adjustments:
the probability of peace is differentiable along the information change,
and any newly emerging mass of immediate concessions is too small to
contribute a first-order boundary term.  The calculation then gives a
precise link to the symmetric result: starting from symmetry, the
one-sided improvement retains exactly half the positive first-order gain
from improving both parties' information.  The following corollary states
this link, with both parties' strategic responses included.

\begin{corollary}[One-sided effects are half of two-sided effects]
\label{onesided-half}
Under the conditions of part~(i) of Theorem \ref{onesided},
$q[F_\eta]=q[F_\eta,F_\eta]$ changes both parties' distributions, whereas
$q(\eta)=q[F_\eta,F]$ changes only party 1's.  Their derivatives satisfy
\begin{equation}\label{onesided-half-identity}
q'(0)=\frac12\left.\frac{d}{d\eta}q[F_\eta]\right|_{\eta=0}.
\end{equation}
\end{corollary}

Two implications deserve attention: the effect of a loss of information,
and the unresolved comparison for changes that are no longer small.

\begin{remark}[Less information near symmetry]\label{rem:onesided-reverse}
Starting from symmetry, a sufficiently small loss of information by one
party lowers the probability of peace.  Indeed, the proof gives the same
positive derivative $q'(0)$ from both sides, so
$q(\eta)<q(0)=q[F]$ for small $\eta<0$, when $F_\eta$ is a mean-preserving
contraction of $F$.  In either direction, how small the change must be
can depend on $\delta$.  The proof also bounds the magnitude of the
first-order response from below in \eqref{onesided-derivative-bound}.
\end{remark}

\begin{remark}[Can a later decline erase the initial gain?]\label{rem:symmetry-benchmark}
Suppose we start from symmetry and improve only one party's information.
Peace initially becomes more likely.  Part~(ii) shows that an improvement
at an asymmetric pair can have the opposite effect, but does not show that
a later decline can outweigh the initial gain.  It remains open whether
\[
q[F^+,F]\geq q[F]
\]
for every admissible mean-preserving spread $F^+$ of $F$, however large.
Our numerical searches have found no counterexample, but we have no proof.
Thus symmetry might minimize the probability of peace \emph{among one-sided
improvements from a given symmetric baseline}; the theorem does not
establish this.
\end{remark}

To illustrate part~(ii), suppose both parties' actual benefits are uniform
on $[0,1]$ and initially fully revealed to their respective holders.
This symmetric starting point leaves no room for further information
improvement for that prior, so we construct a path of \emph{information
loss}.  Keep party 2 fully informed and give party 1 the posterior-value
distribution
\begin{equation}\label{onesided-family}
F_t(x)=(1-t)x+t\frac{x^{20}}{x^{20}+(1-x)^{20}},
\qquad x\in[0,1],\quad 0\leq t\leq1.
\end{equation}
Write $q_t=q[F_t,F_0]$.  At $t=0$ both parties are uniform; at $t=1$,
party 1's posterior values are concentrated around $1/2$, while party 2
remains uniform.  Increasing $t$ is a mean-preserving contraction and can
be implemented by garbling party 1's information.  Information improvements
therefore run in the opposite direction, toward $t=0$.

\begin{figure}[H]
\centering
\includegraphics[width=.94\textwidth]{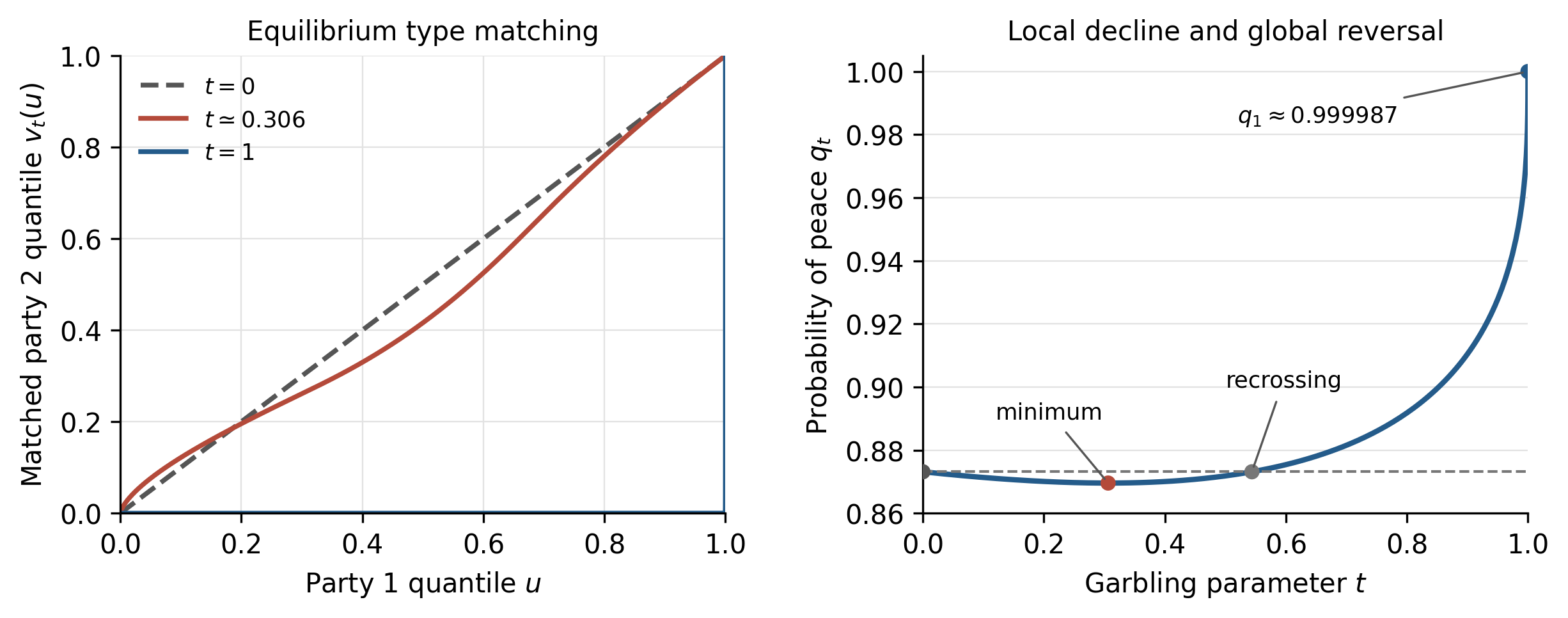}
\caption{Rank matching (left) and peace (right) along
\eqref{onesided-family}.  The rank $v_t(u)$ of party 2 concedes at the
same level as party-1 rank $u$.  Increasing $t$ garbles party 1's
information: peace first falls, then rises.  Information improvements
run from right to left.}
\label{fig:onesided-garbling}
\end{figure}

The right panel of Figure~\ref{fig:onesided-garbling} shows that a little
garbling lowers peace, as Remark~\ref{rem:onesided-reverse} predicts;
further garbling eventually raises it.  Reading the rising portion from right to left gives the adverse
local effect in part~(ii): even a small information improvement can lower
peace away from symmetry.  The Appendix proves the reversal analytically;
the figure illustrates the intervening non-monotonicity.  This does not
settle Remark
\ref{rem:symmetry-benchmark}: moving away from symmetry here means losing,
not gaining, information.

The left panel of Figure~\ref{fig:onesided-garbling} explains the reversal.  At symmetry, equal ranks concede
at equal levels: $v_0(u)=u$.  As garbling increases, the matching changes
until most party-1 ranks are paired with very low ranks of party 2.
Party 1 then concedes at very low hazard levels against most opponent
types.  This early concession accounts for $q_1\approx0.999987$: peace is
nearly certain because confrontation ends quickly, not because prolonged
escalation has become safe.  Only party 1's information deteriorates;
if \emph{both} parties instead had constant posterior value $1/2$, the
symmetric no-information benchmark would be $0.8$ (Corollary
\ref{maxminsp}).  The Supplementary Appendix provides the numerical hazard
levels and checks the incentives for rapid concession.

The uniform starting point in this example satisfies the theorem's two
bounds: $f=1$ and the garbling direction $\delta(x)=F_1'(x)-1$ is
bounded, while the lower-endpoint bound holds with $\alpha=0$.
Reversing the direction gives the local information improvement in
part~(i).  More generally, the theorem permits changes at both ends and a
vanishing upper-endpoint density; it requires neither differentiability
of $f$ or $\delta$ nor a monotone hazard rate or unimodality.
The Appendix derives the asymmetric equilibrium and the first-order
identity, including the positive-kernel formula \eqref{onesided-kernel};
the Supplementary Appendix expands the endpoint arguments.

The one-sided comparison therefore extends the favorable effect locally
at symmetry, but not to every asymmetric starting point.

We make no one-sided welfare claim: information also changes who prevails
and the value of prevailing, so the change in peace alone cannot rank
either party's utility.

\section{Discussion}\label{sec:discussion}

Three things could in principle be changed to make a confrontation more
likely to end before disaster: the escalation technology, the stakes, and how
accurately each party assesses its own benefit from prevailing.  These
three changes have different effects.  The first is neutral, by a logic already known
from earlier work.
The second moves the probability of peace in the expected direction, but
carries no general welfare ranking.  Only the third, in the symmetric comparison, raises the
probability of peace and both parties' ex-ante equilibrium utilities at once.

The neutrality logic has a close precedent.  \citet{N1} showed
that equilibrium damage in a continuous brinkmanship game does not depend on
the technology mapping escalation into risk, and Observation \ref{indie}
records the counterpart for the probability of peace.  The reason is exact
adjustment.  A change in $G$ moves every type's
physical concession level, but it leaves unchanged the accumulated disaster
hazard that type is willing to bear.  Hotlines, wider buffer zones, better
early warning, and tighter control over launch authority may all be valuable
through channels outside the model, including by changing the consequences of
disaster.  What they cannot be presumed to do, insofar as they operate only
through $G$, is raise the probability of peace.

The stakes work in the direction one would expect.  Parties with less to gain
from prevailing are less persistent, so a first-order stochastic decrease in
posterior values raises the probability of peace (Proposition \ref{FOD}).
Measures that reduce the value of prevailing---sanctions triggered by gains taken
by force, or commitments not to recognize such gains---therefore point in the
right direction.  Two things limit how much can be made of this.  A change in
the stakes moves the value of prevailing and planned exposure to disaster
together, and the paper establishes no general ranking of the parties'
utilities: Proposition \ref{fod-welfare} shows only that the exposure effect
can dominate, so that even a first-order increase in benefits can leave both
parties worse off.  The effect of the size of the benefit on welfare is
therefore ambiguous where its effect on peace is not.  And reducing what two
parties are fighting over is easier to recommend than to arrange.

The third instrument is different in kind.  It changes neither the escalation
technology nor the level of the benefit at stake, but only how well each party
understands what prevailing would be worth to it: what the contested
territory, disputed strait, political office, or legal precedent is actually
worth once the costs of acquiring and holding it have been counted.
Systematic estimates of the economic costs of conflict are increasingly
available \citep{Federle}, and can inform the private self-assessment studied
here.  Theorems \ref{mps} and \ref{welfare} say that improving such
self-assessment, while holding the mean fixed, raises both the probability of
peace and each party's ex-ante equilibrium utility.  The distinction is
between changing the stakes and improving their assessment while holding
the mean fixed: only the latter yields both rankings here.
The magnitudes need not be small.  In the uniform illustration with the
benefit range $[0,10]$, moving from no information to full information raises
the probability of peace from $0.286$ to $0.533$.

The mechanism is not greater caution: each party's individual expected
planned exposure remains unchanged.  What falls is expected planned exposure
at the first concession, as shown in Section~\ref{sec:information}.  The
confrontation ends earlier in hazard units on average without either party
planning to bear less hazard on average.

This is not the familiar transparency argument.  In rationalist accounts of
war, conflict can arise because parties misjudge one another and have reason
to misrepresent, so revealing information to the opponent can help prevent it
\citep{Fearon1995}.  Here nothing is revealed to the opponent, and because the
benefits are independent, a more accurate self-assessment tells a party
nothing about the party it faces.  Nor does the result argue against strategic
ambiguity in the tradition of \citet{Schelling}: ambiguity about resolve,
commitment, or the risk of escalation is a different instrument, on which the
model is silent.  The narrower conclusion concerns symmetric improvements in private
self-assessment, not disclosure to an opponent or the removal of strategic
ambiguity.

Although we do not model a sender choosing information, the symmetric
ranking gives a benchmark for anyone seeking to promote peace.  Among the symmetric information structures available given the
parties' prior, full information is the convex-order maximum, so Theorems
\ref{mps} and \ref{welfare} make it optimal on both counts.  The prescription
depends on the objective: a sender maximizing auction revenue, equivalently
expected dissipation, generally does better by pooling \citep{S_W_N}, whereas
\citet{Szydlowski} studies public information about relative strength for a
sender who wishes to sustain conflict.  In this symmetric comparison, the peace objective and the parties' welfare
also align: the same improvement raises both parties' ex-ante equilibrium
utilities, so neither has an ex-ante reason to refuse it.

This benchmark depends on what is learned, who learns it, and when.
The restriction that matters most is the private-value structure.  The results
concern a party's information about its own benefit from prevailing, with the
two benefits independent.  Information about a component the parties
share---how much the contested field actually holds---is a different object;
we make no claim for interdependent-value environments, where even the
equilibrium characterization becomes substantially more delicate.  The
remaining limits are narrower.  When only one party becomes better informed,
the positive result is local (Theorem \ref{onesided}(i)); even a small improvement
from an asymmetric starting point need not help (Theorem \ref{onesided}(ii)); and
the local result cannot be iterated,
because after the first change the parties no longer have the same
distribution.  The
comparison is also between information environments rather than between
moments within a crisis: parties who expect a better assessment to arrive may
wait for it, which is a different game.

Within these limits, the gains extend beyond the two parties.  Escalation
involving energy producers, infrastructure, or transport routes can impose
costs on outsiders through prices and trade \citep{IEA2022,Federle}.  Those
costs are external to our model; an additional loss incurred at disaster
would strengthen the symmetric case for better self-assessment, since both
the parties' welfare and the outside risk would improve.  The same logic
applies to the other attrition settings described in Section~\ref{sec:introduction},
subject to the same private-value and symmetry conditions.

\bibliographystyle{elsarticle-harv}
\bibliography{references}

@article{Acton,
  author  = {Acton, James M.},
  title   = {Escalation through Entanglement: How the Vulnerability of Command-and-Control Systems Raises the Risks of an Inadvertent Nuclear War},
  journal = {International Security},
  year    = {2018},
  volume  = {43},
  number  = {1},
  pages   = {56--99},
}

@article{Basak,
  author  = {Basak, Deepal},
  title   = {Social Value of Public Information in Bargaining},
  journal = {Economic Journal},
  year    = {2024},
  volume  = {134},
  number  = {660},
  pages   = {1356--1378},
}

@article{BergemannMorris,
  author  = {Bergemann, Dirk and Morris, Stephen},
  title   = {Information Design: A Unified Perspective},
  journal = {Journal of Economic Literature},
  year    = {2019},
  volume  = {57},
  number  = {1},
  pages   = {44--95},
}

@article{S_W_N,
  author  = {Bergemann, Dirk and Heumann, Tibor and Morris, Stephen and Sorokin, Constantine and Winter, Eyal},
  title   = {Optimal Information Disclosure in Classic Auctions},
  journal = {American Economic Review: Insights},
  year    = {2022},
  volume  = {4},
  number  = {3},
  pages   = {371--388},
}

@article{Blackwell,
  author  = {Blackwell, David},
  title   = {Equivalent Comparisons of Experiments},
  journal = {Annals of Mathematical Statistics},
  year    = {1953},
  volume  = {24},
  pages   = {265--272},
}

@article{BulowKlemperer,
  author  = {Bulow, Jeremy and Klemperer, Paul},
  title   = {The Generalized War of Attrition},
  journal = {American Economic Review},
  year    = {1999},
  volume  = {89},
  number  = {1},
  pages   = {175--189},
}

@misc{CastilloMiranda,
  author       = {Castillo-Quintana, Martin and Miranda-Romero, Gianfranco},
  title        = {Multiplicity of Equilibria in the War of Attrition with Two-Sided Asymmetric Information},
  year         = {2026},
  howpublished = {arXiv:2603.13634},
}

@article{Chen2025,
  author  = {Chen, Zhuoqiong},
  title   = {Know Thy Enemy: Information Acquisition in Contests},
  journal = {European Economic Review},
  year    = {2025},
  volume  = {177},
  pages   = {105051},
}

@article{ChiMurtoValimaki,
  author  = {Chi, Chang Koo and Murto, Pauli and V\"alim\"aki, Juuso},
  title   = {All-Pay Auctions with Affiliated Binary Signals},
  journal = {Journal of Economic Theory},
  year    = {2019},
  volume  = {179},
  pages   = {99--130},
}

@article{EkmekciZhang,
  author  = {Ekmekci, Mehmet and Zhang, Hanzhe},
  title   = {Reputational Bargaining with External Resolution Opportunities},
  journal = {Review of Economic Studies},
  year    = {2025},
  volume  = {92},
  number  = {4},
  pages   = {2472--2501},
}

@article{EwerhartLareida,
  author  = {Ewerhart, Christian and Lareida, Julia},
  title   = {Voluntary Disclosure in Asymmetric Contests},
  journal = {Review of Economic Studies},
  year    = {2024},
  volume  = {91},
  number  = {6},
  pages   = {3402--3422},
}

@article{FanningDeadline,
  author  = {Fanning, Jack},
  title   = {Reputational Bargaining and Deadlines},
  journal = {Econometrica},
  year    = {2016},
  volume  = {84},
  number  = {3},
  pages   = {1131--1179},
}

@article{Federle,
  author  = {Federle, Jonathan and Meier, Andr\'e and M\"uller, Gernot J. and Mutschler, Willi and Schularick, Moritz},
  title   = {The Price of War},
  journal = {American Economic Review},
  year    = {2026},
  volume  = {116},
  number  = {3},
  pages   = {791--827},
}

@article{Fearon1995,
  author  = {Fearon, James D.},
  title   = {Rationalist Explanations for War},
  journal = {International Organization},
  year    = {1995},
  volume  = {49},
  number  = {3},
  pages   = {379--414},
}

@article{FengSong2025,
  author  = {Feng, Xin and Song, Shuangteng},
  title   = {Information Acquisition in All-Pay Contests},
  journal = {Economics Letters},
  year    = {2025},
  volume  = {256},
  pages   = {112585},
}

@article{FudTir,
  author  = {Fudenberg, Drew and Tirole, Jean},
  title   = {A Theory of Exit in Duopoly},
  journal = {Econometrica},
  year    = {1986},
  volume  = {54},
  number  = {4},
  pages   = {943--960},
}

@article{FuchsSkrzypacz,
  author  = {Fuchs, William and Skrzypacz, Andrzej},
  title   = {Bargaining with Deadlines and Private Information},
  journal = {American Economic Journal: Microeconomics},
  year    = {2013},
  volume  = {5},
  number  = {4},
  pages   = {219--243},
}

@article{GanuzaPenalva,
  author  = {Ganuza, Juan-Jos\'e and Penalva, Jos\'e S.},
  title   = {Signal Orderings Based on Dispersion and the Supply of Private Information in Auctions},
  journal = {Econometrica},
  year    = {2010},
  volume  = {78},
  number  = {3},
  pages   = {1007--1030},
}

@article{Georgiadis,
  author  = {Georgiadis, George and Kim, Youngsoo and Kwon, H. Dharma},
  title   = {The Absence of Attrition in a War of Attrition under Complete Information},
  journal = {Games and Economic Behavior},
  year    = {2022},
  volume  = {131},
  pages   = {171--185},
}

@article{Gieczewski,
  author  = {Gieczewski, Germ\'an},
  title   = {Evolving Wars of Attrition},
  journal = {Journal of Economic Theory},
  year    = {2025},
  volume  = {224},
  pages   = {105967},
}

@article{Hennigs,
  author  = {Hennigs, Raphaela},
  title   = {Conflict Prevention by {Bayesian} Persuasion},
  journal = {Journal of Public Economic Theory},
  year    = {2021},
  volume  = {23},
  pages   = {710--731},
}

@article{Herrera,
  author  = {Herrera, Helios and Mac\'e, Antonin and N\'u\~nez, Mat\'ias},
  title   = {Political Brinkmanship and Compromise},
  journal = {International Economic Review},
  year    = {2025},
  volume  = {66},
  number  = {3},
  pages   = {1317--1339},
}

@article{Horner,
  author  = {H\"orner, Johannes and Morelli, Massimo and Squintani, Francesco},
  title   = {Mediation and Peace},
  journal = {Review of Economic Studies},
  year    = {2015},
  volume  = {82},
  number  = {4},
  pages   = {1483--1501},
}

@book{IEA2022,
  author    = {{International Energy Agency}},
  title     = {World Energy Outlook 2022},
  publisher = {OECD Publishing},
  address   = {Paris},
  year      = {2022},
}

@article{Distr,
  author  = {Jensen, Martin K.},
  title   = {Distributional Comparative Statics},
  journal = {Review of Economic Studies},
  year    = {2018},
  volume  = {85},
  number  = {1},
  pages   = {581--610},
}

@article{BPers,
  author  = {Kamenica, Emir and Gentzkow, Matthew},
  title   = {{Bayesian} Persuasion},
  journal = {American Economic Review},
  year    = {2011},
  volume  = {101},
  number  = {6},
  pages   = {2590--2615},
}

@article{IA_woa,
  author  = {Kim, Kyungmin and Lee, Frances Zhiyun Xu},
  title   = {Information Acquisition in a War of Attrition},
  journal = {American Economic Journal: Microeconomics},
  year    = {2014},
  volume  = {6},
  number  = {2},
  pages   = {37--78},
}

@article{KovenockMorathMuenster,
  author  = {Kovenock, Dan and Morath, Florian and M\"unster, Johannes},
  title   = {Information Sharing in Contests},
  journal = {Journal of Economics \& Management Strategy},
  year    = {2015},
  volume  = {24},
  number  = {3},
  pages   = {570--596},
}

@article{KotowskiLi,
  author  = {Kotowski, Maciej H. and Li, Fei},
  title   = {The War of Attrition and the Revelation of Valuable Information},
  journal = {Economics Letters},
  year    = {2014},
  volume  = {124},
  number  = {3},
  pages   = {420--423},
}

@article{w-o-a-all-pay,
  author  = {Krishna, Vijay and Morgan, John},
  title   = {An Analysis of the War of Attrition and the All-Pay Auction},
  journal = {Journal of Economic Theory},
  year    = {1997},
  volume  = {72},
  number  = {2},
  pages   = {343--362},
}

@article{LuParreiras,
  author  = {Lu, Jingfeng and Parreiras, S\'ergio O.},
  title   = {Monotone Equilibrium of Two-Bidder All-Pay Auctions Redux},
  journal = {Games and Economic Behavior},
  year    = {2017},
  volume  = {104},
  pages   = {78--91},
}

@article{MaManove,
  author  = {Ma, Ching-to Albert and Manove, Michael},
  title   = {Bargaining with Deadlines and Imperfect Player Control},
  journal = {Econometrica},
  year    = {1993},
  volume  = {61},
  number  = {6},
  pages   = {1313--1339},
}

@article{MM,
  author  = {Marschak, Jacob and Miyasawa, Koichi},
  title   = {Economic Comparability of Information Systems},
  journal = {International Economic Review},
  year    = {1968},
  volume  = {9},
  number  = {2},
  pages   = {137--174},
}

@article{MorathIgnorance,
  author  = {Morath, Florian},
  title   = {Volunteering and the Strategic Value of Ignorance},
  journal = {Social Choice and Welfare},
  year    = {2013},
  volume  = {41},
  number  = {1},
  pages   = {99--131},
}

@article{MorathMuenster2008,
  author  = {Morath, Florian and M\"unster, Johannes},
  title   = {Private versus Complete Information in Auctions},
  journal = {Economics Letters},
  year    = {2008},
  volume  = {101},
  number  = {3},
  pages   = {214--216},
}

@article{MorathMuenster2013,
  author  = {Morath, Florian and M\"unster, Johannes},
  title   = {Information Acquisition in Conflicts},
  journal = {Economic Theory},
  year    = {2013},
  volume  = {54},
  number  = {1},
  pages   = {99--129},
}

@article{Myatt,
  author  = {Myatt, David P.},
  title   = {The Impact of Perceived Strength in the War of Attrition},
  journal = {Games and Economic Behavior},
  year    = {2025},
  volume  = {150},
  pages   = {260--277},
}

@article{N1,
  author  = {Nalebuff, Barry},
  title   = {Brinkmanship and Nuclear Deterrence: The Neutrality of Escalation},
  journal = {Conflict Management and Peace Science},
  year    = {1986},
  volume  = {9},
  number  = {2},
  pages   = {19--30},
}

@article{NalebuffRiley,
  author  = {Nalebuff, Barry and Riley, John},
  title   = {Asymmetric Equilibria in the War of Attrition},
  journal = {Journal of Theoretical Biology},
  year    = {1985},
  volume  = {113},
  number  = {3},
  pages   = {517--527},
}

@article{Clara,
  author  = {Ponsat\'i, Clara and S\'akovics, J\'ozsef},
  title   = {The War of Attrition with Incomplete Information},
  journal = {Mathematical Social Sciences},
  year    = {1995},
  volume  = {29},
  number  = {3},
  pages   = {239--254},
}

@article{Posen,
  author  = {Posen, Barry R.},
  title   = {Inadvertent Nuclear War? Escalation and {NATO}'s Northern Flank},
  journal = {International Security},
  year    = {1982},
  volume  = {7},
  number  = {2},
  pages   = {28--54},
}

@article{Powell1988,
  author  = {Powell, Robert},
  title   = {Nuclear Brinkmanship with Two-Sided Incomplete Information},
  journal = {American Political Science Review},
  year    = {1988},
  volume  = {82},
  number  = {1},
  pages   = {155--178},
}

@article{RentschlerTurocy,
  author  = {Rentschler, Lucas and Turocy, Theodore L.},
  title   = {Two-Bidder All-Pay Auctions with Interdependent Valuations, Including the Highly Competitive Case},
  journal = {Journal of Economic Theory},
  year    = {2016},
  volume  = {163},
  pages   = {435--466},
}

@book{Schelling,
  author    = {Schelling, Thomas C.},
  title     = {The Strategy of Conflict},
  publisher = {Harvard University Press},
  address   = {Cambridge, MA},
  year      = {1960}
}

@article{SchramChance,
  author  = {Schram, Peter},
  title   = {Conflicts that Leave Something to Chance},
  journal = {International Organization},
  year    = {2025},
  volume  = {79},
  number  = {2},
  pages   = {199--232},
}

@article{Seel,
  author  = {Seel, Christian},
  title   = {The Value of Information in Asymmetric All-Pay Auctions},
  journal = {Games and Economic Behavior},
  year    = {2014},
  volume  = {86},
  pages   = {330--338},
}

@book{SOrders,
  author    = {Shaked, Moshe and Shanthikumar, J. George},
  title     = {Stochastic Orders},
  publisher = {Springer},
  address   = {New York},
  year      = {2007},
}

@article{Siegel2014,
  author  = {Siegel, Ron},
  title   = {Asymmetric All-Pay Auctions with Interdependent Valuations},
  journal = {Journal of Economic Theory},
  year    = {2014},
  volume  = {153},
  pages   = {684--702},
}

@article{Smith,
  author  = {Maynard Smith, John and Parker, Geoffrey A.},
  title   = {The Logic of Asymmetric Contests},
  journal = {Animal Behaviour},
  year    = {1976},
  volume  = {24},
  number  = {1},
  pages   = {159--175},
}

@article{Szydlowski,
  author  = {Szydlowski, Martin},
  title   = {Fomenting Conflict},
  journal = {Journal of Economic Theory},
  year    = {2024},
  volume  = {220},
  pages   = {105875},
}

@article{Warneryd,
  author  = {W\"arneryd, Karl},
  title   = {Information in Conflicts},
  journal = {Journal of Economic Theory},
  year    = {2003},
  volume  = {110},
  number  = {1},
  pages   = {121--136},
}

@techreport{Wan,
  author      = {Wan, Wilfred},
  title       = {Addressing Multidomain Nuclear Escalation Risk},
  institution = {Stockholm International Peace Research Institute},
  address     = {Stockholm},
  type        = {SIPRI Research Policy Paper},
  year        = {2026},
}

@article{Zheng,
  author  = {Zheng, Charles Z.},
  title   = {Necessary and Sufficient Conditions for Peace: Implementability versus Security},
  journal = {Journal of Economic Theory},
  year    = {2019},
  volume  = {180},
  pages   = {135--166},
}

\pagebreak

\section*{Appendix}

The proofs follow the order of the main text.  We first establish the
selected symmetric equilibrium and its probability of peace
(Section~\ref{sec:model}), then compare stakes, information, and welfare
(Section~\ref{sec:stakes}).  Finally, we turn to one-sided information
(Section~\ref{sec:onesided}), proving the local result and half-effect
identity before the reversal away from symmetry.  Auxiliary lemmas are
introduced at the steps that require them.  The Supplementary Appendix
provides expanded derivations of the equilibrium-selection and endpoint
arguments.

\subsection*{Model and equilibrium}

Proposition \ref{FOC} requires more than constructing a symmetric
strategy: we must establish that the common wild-type refinement selects
it.  We first express concession decisions in units of disaster hazard,
then identify the condition that makes the perturbed equilibrium unique.

Let
\[
H(t)=-\log S_G(t)
\]
be accumulated disaster hazard.  Because $G$ has a positive density and full
support, $H$ is a strictly increasing map from $[0,\infty)$ onto itself.
If a party plans to concede at hazard level $b=H(t)$, reaching that level has
probability $e^{-b}$.  After integrating out the disaster threshold, the game
on the $b$-scale is therefore an exponentially discounted concession game:
the conceding party obtains $e^{-b}$, the other party obtains
$(1+\theta)e^{-b}$, and perpetual disagreement gives zero.  This is the
attrition structure analyzed by \citet{Clara}.

The remaining selection issue is that identical posterior-value
distributions need not, by themselves, rule out asymmetric equilibria.
Give each party probability $c\in(0,1)$ of being an ordinary type and
probability $1-c$ of being a wild type that never concedes.
We use the standard no-atoms and monotonicity characterization of
\citet{NalebuffRiley} and \citet{Clara}.  Within that characterization,
the following lemma shows how the common wild-type perturbation fixes
both concession strategies.  We state it for possibly different
ordinary-type distributions, so that it can also be used in the
one-sided comparison later; its symmetric case is what is needed here.

\begin{lemma}[The common-upper-endpoint condition]\label{wild-type-uniqueness}
Fix $c\in(0,1)$ and let $Q_i$ be party $i$'s ordinary-type quantile function.
Within the standard monotone characterization, with no atoms at positive
concession levels and at most one party having an atom at zero, the wild-type
perturbation has a unique equilibrium.  The ordinary types' finite
concession ranges have a common upper endpoint, and ranks $u$ and $v$
assigned the same positive concession level satisfy
\begin{equation}\label{wild-type-matching}
\Lambda_{2,c}(v)=\Lambda_{1,c}(u)+C_c,
\qquad
C_c=\Lambda_{2,c}(1)-\Lambda_{1,c}(1),
\end{equation}
where
\[
\Lambda_{i,c}(w)=\int_{1/2}^w\frac{ds}{Q_i(s)(1-cs)}.
\]
Thus the rank-matching function, any interval of lowest types that concedes
immediately, and the positive parts of both concession strategies are uniquely
determined.  If
$Q_1=Q_2$, then $C_c=0$ and the equilibrium is symmetric.
\end{lemma}

\begin{proof}{of Lemma}{\ref{wild-type-uniqueness}}
Suppose one party's highest ordinary concession level were strictly above the
other's.  Upon reaching the lower endpoint without a concession, an ordinary
type would know that the opponent is wild: the ordinary opponent has already
exhausted its concession range.  Continuing would then add disaster risk
without any chance of prevailing, so the higher endpoint cannot be optimal.  The
two finite concession ranges must therefore share an upper endpoint.

On their common positive range, divide the two first-order conditions and let
$v=v(u)$ denote the rank of party 2 assigned the same concession level as
party 1's rank $u$.
This gives
\[
\frac{v'(u)}{Q_2(v(u))[1-cv(u)]}
=\frac{1}{Q_1(u)(1-cu)}.
\]
Integration yields the first equation in \eqref{wild-type-matching}.  At the
common upper endpoint the top ordinary types are matched, so $u=v=1$ and the
second equation follows.  The terminal values are finite because
$1-cs\geq1-c>0$ and the quantile functions are bounded away from zero near
the top.

Since each $\Lambda_{i,c}$ is strictly increasing, \eqref{wild-type-matching}
determines a unique rank-matching function wherever its right-hand side lies
in the other party's range.  If it falls outside that range at the bottom, the uniquely
determined residual interval consists of lowest types that concede at level
zero.  The lowest escalating type of each party concedes at level zero: if
the lowest positive concession level were $b_0>0$, no ordinary opponent
would concede on $(0,b_0)$, so a type planning to concede at $b_0$ would gain
by conceding at zero and avoiding the disaster risk on $(0,b_0)$.  Finally,
if $b_i$ denotes party $i$'s planned-exposure function in rank coordinates,
the first-order conditions give
\[
b_1'(u)=\frac{cQ_2(v(u))}{1-cu},
\qquad
b_2'(v)=\frac{cQ_1(u(v))}{1-cv}.
\]
Integrating from the lowest positive concession level, which is zero, pins
down both planned-exposure functions and therefore both physical concession
strategies.  Hence there is at most one equilibrium within the standard
characterization.

For existence, construct these strategies from the matching equation.
If party 1 with value $\theta$ chooses the positive level assigned to its
rank $u$, differentiating its payoff gives
\[
\frac{dU_{1,c}}{du}
=c e^{-b_1(u)}v'(u)\bigl[\theta-Q_1(u)\bigr].
\]
Thus its prescribed rank is a global best response on the positive range;
the same argument applies to party 2.  An unmatched bottom type prefers
zero to every positive level.  If instead the opponent has an atom at
zero, conceding at zero yields no more than the limit payoff from conceding
just after zero, because the latter wins rather than ties against that atom.
Finally, continuing beyond the common top only adds disaster risk against
a wild opponent.  These observations cover all deviations and establish
existence.  When $Q_1=Q_2$, equation
\eqref{wild-type-matching} gives $C_c=0$ and $v(u)=u$, establishing symmetry.
\end{proof}

\begin{proof}{of Proposition}{\ref{FOC}}
Specialize Lemma \ref{wild-type-uniqueness} to the case in which both
parties' ordinary types have distribution $F$.  The perturbed equilibrium
is unique and, by symmetry, is symmetric and increasing among ordinary
types.  Write its physical concession strategy as $\tau_c$.

If an ordinary type $\theta$ chooses the level assigned to ordinary type $x$,
its expected payoff is
\begin{equation}\label{perturbed-EU}
U_c(x,\theta;\tau_c)
=S_G(\tau_c(x))[1-cF(x)]
+(1+\theta)c\int_{\underline\theta}^x
S_G(\tau_c(y))\,dF(y).
\end{equation}
The first term includes both wild opponents and ordinary opponents with type
above $x$.  Differentiating gives
\begin{equation}\label{perturbed-report-derivative}
\frac{\partial U_c}{\partial x}
=-g(\tau_c(x))\tau_c'(x)[1-cF(x)]
+\theta cS_G(\tau_c(x))f(x).
\end{equation}
Truthful choice by type $x$ therefore requires
\begin{equation}\label{perturbed-FOC}
\lambda_G(\tau_c(x))\tau_c'(x)
=\frac{xc f(x)}{1-cF(x)}.
\end{equation}
The lowest ordinary type concedes immediately: waiting would reduce its
payoff while it concedes before every ordinary opponent.  Hence
$\tau_c(\underline\theta)=0$.  Define
\[
\beta_{F,c}(x)
=\int_{\underline\theta}^x
\frac{zc f(z)}{1-cF(z)}\,dz.
\]
Integrating \eqref{perturbed-FOC} gives
\begin{equation}\label{perturbed-survival}
S_G(\tau_c(x))=e^{-\beta_{F,c}(x)}.
\end{equation}
Conversely, substituting \eqref{perturbed-FOC} into
\eqref{perturbed-report-derivative} yields
\[
\frac{\partial U_c}{\partial x}
=cS_G(\tau_c(x))f(x)(\theta-x).
\]
Thus the constructed strategy is globally incentive compatible.  Deviations
above the top level $\tau_c(\bar\theta)$ are dominated: beyond that level any
opponent still present is wild, so further escalation adds disaster risk
without a chance of prevailing.  A deviation to level zero is already included
as the report $x=\underline\theta$.  Notice also that
$1-cF(x)\geq1-c$, so $\beta_{F,c}(\bar\theta)<\infty$: every
ordinary type concedes at a finite level, as the wild-type refinement
requires.

Finally let $c\uparrow1$.  Since
$\partial_c[c/(1-cF(z))]=(1-cF(z))^{-2}>0$, monotone convergence gives, for
every $x<\bar\theta$,
\[
\beta_{F,c}(x)\longrightarrow
\int_{\underline\theta}^x\frac{zf(z)}{S_F(z)}\,dz
=\beta_F(x).
\]
By \eqref{perturbed-survival},
$\tau_c=H^{-1}\circ\beta_{F,c}$, so the selected physical concession strategies converge
pointwise on the interior to $\tau=H^{-1}\circ\beta_F$.  This proves
\eqref{FOC-i} and \eqref{equilibrium-survival}.  Moreover,
$\beta_F(x)\to\infty$ as $x\uparrow\bar\theta$: for any
$x_0\in(\underline\theta,\bar\theta)$ with $x_0>0$,
\[
\int_{x_0}^x\frac{zf(z)}{S_F(z)}\,dz
\geq x_0\log\frac{S_F(x_0)}{S_F(x)}\longrightarrow\infty.
\]
Thus every finite concession level is represented by some
$x<\bar\theta$; never conceding is the limit as $x\uparrow\bar\theta$.
Taking $c=1$ in
\eqref{perturbed-EU} and repeating the calculation in
\eqref{perturbed-report-derivative} gives
\begin{equation}\label{report-derivative}
U_x(x,\theta;\tau)
=S_G(\tau(x))f(x)(\theta-x).
\end{equation}
Hence the limiting physical concession strategy is strictly increasing, globally incentive
compatible, and uniquely determined by the refinement; the global
best-response property follows directly from the sign in
\eqref{report-derivative}.
\end{proof}

With the selected strategy established, we can derive the probability of
peace used in the remaining proofs.  Let $\nu(u)=F^{-1}(u)$ and
$\theta_{(1)}=\min\{\theta_1,\theta_2\}$.  The density of
$\theta_{(1)}$ is $2S_F(x)f(x)$, so
\[
q[F]
=E[e^{-\beta_F(\theta_{(1)})}]
=2\int_{\underline\theta}^{\bar\theta}
e^{-\beta_F(x)}S_F(x)f(x)\,dx.
\]
Under the change of variables $u=F(x)$,
\[
\beta_F(\nu(u))
=\int_0^u\frac{\nu(v)}{1-v}\,dv.
\]
Substitution gives the functional in \eqref{peace1}.

\subsection*{Stakes and information}

With the selected equilibrium in hand, we can compare distributions using
the probability-of-peace formula \eqref{peace1}.  As in the main text,
we first raise the stakes and then improve information while holding the
mean fixed.  The utility comparisons require a separate expression for
ex-ante payoffs.

\subsubsection*{Higher stakes}

First-order stochastic dominance orders the quantiles pointwise, so the
peace comparison is immediate.  To assess whether the higher stakes also
benefit the parties, we first express their equilibrium utility in the
same quantile coordinates.

For the utility comparison, write $\overline U[\nu]=\overline U[F]$ for either
party's ex-ante equilibrium utility and define
\[
\sigma_\nu(u)
=\exp\left[-\int_0^u\frac{\nu(v)}{1-v}\,dv\right],
\qquad
N_\nu(u)=\int_u^1\nu(v)\,dv.
\]
Changing variables in \eqref{EU}, averaging over a party's own type, and
using Fubini's theorem gives
\begin{align*}
\overline U[\nu]
&=\int_0^1\left[\sigma_\nu(u)(1-u)
+(1+\nu(u))\int_0^u\sigma_\nu(v)\,dv\right]du\\
&=\int_0^1\sigma_\nu(v)
\left[(1-v)+\int_v^1(1+\nu(u))\,du\right]dv.
\end{align*}
Therefore
\begin{equation}\label{utility-quantile}
\overline U[\nu]
=\int_0^1\sigma_\nu(u)
\left[2(1-u)+N_\nu(u)\right]du.
\end{equation}

\begin{proof}{of Proposition}{\ref{FOD}}
If $F^1$ first-order stochastically dominates $F^0$, then the
corresponding quantile functions satisfy $\nu^1(u)\geq\nu^0(u)$ for
all $u\in[0,1]$.  The exponential in \eqref{peace1} is therefore weakly lower
under $F^1$ at every quantile rank $u$, proving the result.  The inequality
is strict whenever the dominance comparison is strict on a set of positive
measure.
\end{proof}

\begin{proof}{of Proposition}{\ref{fod-welfare}}
Use the quantile representation \eqref{utility-quantile} derived above.

For $F^0(x)=x$, the quantile is $\nu^0(u)=u$, so
\[
\sigma_{\nu^0}(u)=e^u(1-u),
\qquad
N_{\nu^0}(u)=\frac{1-u^2}{2}.
\]
Equations \eqref{peace1} and \eqref{utility-quantile} give
\[
q[F^0]=4e-10\approx0.873127,
\qquad
\overline U[F^0]=3e-7\approx1.154845.
\]

For $F^1(x)=x^2$, the quantile is $\nu^1(u)=\sqrt u$.  Direct integration
gives
\[
\sigma_{\nu^1}(u)
=e^{2\sqrt u}\frac{1-\sqrt u}{1+\sqrt u},
\qquad
N_{\nu^1}(u)=\frac23(1-u^{3/2}).
\]
Substitution yields
\[
q[F^1]=\frac{9-e^2}{2}\approx0.805472,
\qquad
\overline U[F^1]\approx1.119844.
\]
Thus both the probability of peace and ex-ante utility fall under this
first-order stochastic increase.
\end{proof}

\subsubsection*{More accurate information}

For the symmetric information comparison, we first prove the more general
planned-exposure ordering in Proposition \ref{mps-dcx}; Theorem \ref{mps} then follows
by choosing the exponential survival transform.

\begin{proof}{of Proposition}{\ref{mps-dcx}}
For a distribution $F$ with quantile function $\nu=F^{-1}$, let
$V=F(\theta_{(1)})$, which has density $2(1-u)$ on $[0,1]$.  Then the
random variable $Z_F=\beta_F(\theta_{(1)})$ defined in the main text can be
written as
\[
Z_F=\int_0^V\frac{\nu(u)}{1-u}\,du.
\]
Thus $Z_F$ is the accumulated disaster hazard at the planned concession
level of the lower-value party, and equation \eqref{peace1} gives
$q[F]=E[e^{-Z_F}]$.

For $k=0,1$, let
\[
\nu^k(u)=(F^k)^{-1}(u),
\]
and define the quantile difference and its integral by
\[
\Delta\nu(u)=\nu^1(u)-\nu^0(u),
\qquad
\Gamma(u)=\int_0^u\Delta\nu(v)\,dv.
\]
The quantile characterization of convex order gives
\begin{equation}\label{mps-quantile-condition}
\Gamma(u)\leq0\quad\text{for every }u\in[0,1],
\qquad
\Gamma(0)=\Gamma(1)=0.
\end{equation}
The functional in \eqref{peace1}, and the more general functional below, is
well defined for every bounded nonnegative nondecreasing quantile function.
The path used next is therefore a path of quantile functions; only this
functional representation is used at intermediate points.  Assumption
\ref{ass:regularity} is needed to identify the endpoints with the smooth
equilibria characterized in Proposition \ref{FOC}, not to differentiate
along the path.
Connect the quantiles by
$\nu_\xi=(1-\xi)\nu^0+\xi\nu^1$, $\xi\in[0,1]$, and set
\[
\mathcal Z_\xi(u)=\int_0^u\frac{\nu_\xi(v)}{1-v}\,dv,
\qquad
\mathcal Q_\Phi(\xi)=2\int_0^1(1-u)\Phi(\mathcal Z_\xi(u))\,du.
\]
Since $\Phi$ is convex and decreasing,
$|\Phi'(z)|\leq|\Phi'(0)|$ for $z\geq0$.  Since
$|\Delta\nu|\leq\bar\theta$, the difference quotient of
$\Phi(\mathcal Z_\xi(u))$ in $\xi$ is bounded by
$|\Phi'(0)|\bar\theta\log(1/(1-u))$, and
$(1-u)\log(1/(1-u))$ is integrable.  This justifies differentiation under
the integral and the exchange of integrals below.
If
$k_\xi(u)=-\Phi'(\mathcal Z_\xi(u))$, differentiation under the integral and
Fubini's theorem give
\begin{align*}
\mathcal Q_\Phi'(\xi)
&=-2\int_0^1(1-u)k_\xi(u)
\left(\int_0^u\frac{\Delta\nu(v)}{1-v}\,dv\right)du\\
&=-2\int_0^1 A_\xi(v)\Delta\nu(v)\,dv,
\end{align*}
where
\[
A_\xi(v)=\frac{1}{1-v}\int_v^1(1-u)k_\xi(u)\,du.
\]
Thus $(1-v)A_\xi(v)$ aggregates the contribution from posterior ranks above
$v$.  In the peace comparison, where $\Phi(z)=e^{-z}$, it is one half of
the probability of peace generated by those ranks.
Since $\Delta\nu=\Gamma'$, integration by parts gives
\begin{equation}\label{mps-path-derivative}
\mathcal Q_\Phi'(\xi)=2\int_0^1A_\xi'(v)\Gamma(v)\,dv.
\end{equation}
The boundary term vanishes because $\Gamma(0)=\Gamma(1)=0$ and, by the
bound below,
$0\leq A_\xi(v)\leq(1-v)k_\xi(0)/2\to0$ as $v\uparrow1$.
Because posterior values are nonnegative, $\mathcal Z_\xi$ is nondecreasing.
Since $\Phi$ is decreasing and convex, $k_\xi$ is therefore nonnegative and
decreasing in $u$.  Hence
\[
0\leq\int_v^1(1-u)k_\xi(u)\,du
\leq k_\xi(v)\int_v^1(1-u)\,du
=\frac{(1-v)^2}{2}k_\xi(v),
\]
and therefore
\begin{equation}\label{A-prime-bound}
A_\xi'(v)
=-k_\xi(v)
+\frac{1}{(1-v)^2}\int_v^1(1-u)k_\xi(u)\,du
\leq-\frac12k_\xi(v)\leq0.
\end{equation}
Together with \eqref{mps-quantile-condition}, equation
\eqref{mps-path-derivative} implies $\mathcal Q_\Phi'(\xi)\geq0$.  Integrating
over $\xi\in[0,1]$ proves \eqref{dcx-ordering}.  If $\Phi'<0$ everywhere,
then $k_\xi(v)>0$ and $A_\xi'(v)<0$ everywhere, so the inequality is strict
whenever the convex-order comparison is nontrivial: otherwise
$\Gamma\equiv0$, which would imply $\Delta\nu=0$ almost everywhere.

Finally, writing
$\mathcal Z_\nu(u)=\int_0^u\nu(v)/(1-v)\,dv$ and using the
density $2(1-u)$ of $V$, Fubini's theorem gives
\begin{align*}
E[Z_F]
&=2\int_0^1(1-u)\mathcal Z_\nu(u)\,du\\
&=\int_0^1(1-v)\nu(v)\,dv
=\frac12E[\theta_{(1)}].
\end{align*}
This also verifies directly the linear case $\Phi(z)=-z$.
\end{proof}

\begin{proof}{of Theorem}{\ref{mps}}
Apply Proposition \ref{mps-dcx} with the decreasing convex function
$\Phi(z)=e^{-z}$ and use $q[F]=E[e^{-Z_F}]$.
\end{proof}

\subsubsection*{Equilibrium welfare}

The peace comparison does not by itself rank the parties' utilities:
information also changes the value of the allocation at peaceful
resolution.  We therefore return to the utility formula
\eqref{utility-quantile} and compare it along the same quantile path.

\begin{proof}{of Theorem}{\ref{welfare}}
Conditional on a peaceful resolution, the sum of the two parties'
payoffs is $2+\widetilde\theta_w$, where $w$ denotes the winner.
The identity
\[
E[\widetilde\theta_w\mathbf 1_{\{\mathrm{peace}\}}]
=E[\theta_w\mathbf 1_{\{\mathrm{peace}\}}]
\]
follows from iterated expectations.  Conditional on the parties' signals
and the independent threshold $T$, the winner and the peace indicator are
determined, while
$E[\widetilde\theta_i\mid s_1,s_2,T]
=E[\widetilde\theta_i\mid s_i]=\theta_i$.
Hence aggregate
ex-ante welfare can be written as
\begin{equation}\label{aggregate-welfare}
W[\nu]=2q[\nu]+\mathcal A[\nu],
\qquad
\mathcal A[\nu]=2\int_0^1\sigma_\nu(u)N_\nu(u)\,du.
\end{equation}
Theorem \ref{mps} shows that the first term increases under a
mean-preserving spread.  It remains to show that the allocation component
$\mathcal A$ also increases.

Use the quantile path $\nu_\xi$, the quantile difference $\Delta\nu$, and
its integral $\Gamma$ from the proof of
Proposition \ref{mps-dcx}.  Write $\sigma_\xi=\sigma_{\nu_\xi}$ and
$N_\xi=N_{\nu_\xi}$, and let $\mathcal A(\xi)=\mathcal A[\nu_\xi]$.
Since the two endpoint distributions have the same
mean,
\[
\frac{d}{d\xi}N_\xi(u)
=\int_u^1\Delta\nu(v)\,dv=-\Gamma(u).
\]
Define
\[
\Psi_\xi(v)=\int_v^1\sigma_\xi(u)N_\xi(u)\,du,
\qquad
R_\xi(v)=\frac{\Psi_\xi(v)}{1-v}.
\]
The bounds $N_\xi(u)\leq\bar\theta(1-u)$ and
$|\partial_\xi\sigma_\xi(u)|\leq\bar\theta\log(1/(1-u))$
justify differentiation and Fubini below.  Differentiating gives
\[
\partial_\xi\sigma_\xi(u)
=-\sigma_\xi(u)\int_0^u\frac{\Delta\nu(v)}{1-v}\,dv,
\qquad
\partial_\xi N_\xi(u)=-\Gamma(u).
\]
Hence Fubini's theorem implies
\begin{align*}
\frac12\mathcal A'(\xi)
&=-\int_0^1\frac{\Delta\nu(v)}{1-v}\Psi_\xi(v)\,dv
  -\int_0^1\sigma_\xi(v)\Gamma(v)\,dv\\
&=-\int_0^1R_\xi(v)\Gamma'(v)\,dv
  -\int_0^1\sigma_\xi(v)\Gamma(v)\,dv.
\end{align*}
Integrating the first term by parts, using
$\Gamma(0)=\Gamma(1)=0$ and
$0\leq R_\xi\leq\sigma_\xi N_\xi$, gives
\begin{equation}\label{allocation-derivative}
\frac12\mathcal A'(\xi)
=\int_0^1\bigl[R_\xi'(v)-\sigma_\xi(v)\bigr]\Gamma(v)\,dv.
\end{equation}
Both $\sigma_\xi$ and $N_\xi$ are nonnegative and decreasing.  Therefore
\[
\Psi_\xi(v)
\leq(1-v)\sigma_\xi(v)N_\xi(v),
\]
and
\begin{align*}
R_\xi'(v)-\sigma_\xi(v)
&=-\frac{\sigma_\xi(v)N_\xi(v)}{1-v}
 +\frac{\Psi_\xi(v)}{(1-v)^2}-\sigma_\xi(v)\\
&\leq-\sigma_\xi(v)<0.
\end{align*}
Since $\Gamma\leq0$, equation \eqref{allocation-derivative} implies
$\mathcal A'(\xi)\geq0$, strictly if the convex-order comparison is strict.
Thus both terms in \eqref{aggregate-welfare} increase.  Ex-ante
symmetry gives $\overline U[\nu]=W[\nu]/2$, proving the theorem.
\end{proof}

\subsubsection*{Bounds and magnitude}

The information ordering also identifies the extreme probabilities of
peace at a fixed mean.  The next proof obtains the bounds by comparing
posterior-value distributions in convex order.

\begin{proof}{of Corollary}{\ref{maxminsp}}
Among distributions on $[\underline\theta,\bar\theta]$ with mean $\mu$,
the least spread is the point mass at $\mu$, while the most spread places
probability
$(\bar\theta-\mu)/(\bar\theta-\underline\theta)$ at
$\underline\theta$ and probability
$(\mu-\underline\theta)/(\bar\theta-\underline\theta)$ at
$\bar\theta$.  Although these distributions are not smooth, their quantile
functions are bounded, nonnegative, and nondecreasing.  These are precisely
the properties used in the proof of Proposition \ref{mps-dcx}: they make
$\mathcal Z_\xi$ nondecreasing and $k_\xi$ decreasing.  Thus the proof
applies even though the quantiles are discontinuous; no density or interior
equilibrium characterization is used along the path.  It identifies the
endpoint distributions directly as the minimizer and maximizer of the
continuous extension of \eqref{peace1}.  For regular distributions that
extension equals the equilibrium probability of peace.

For the point mass, the quantile function is identically $\mu$, so
equation \eqref{peace1} gives
\[
q[\nu]=2\int_0^1(1-u)^{1+\mu}\,du
=\frac{2}{2+\mu}.
\]
If $\mu$ equals either support endpoint, the only admissible distribution
is a point mass, and both bounds follow from this calculation.
Otherwise $0<\pi<1$.  For the maximally spread distribution, let $\pi$ be as defined in Corollary
\ref{maxminsp}; it is the probability assigned to $\bar\theta$.  The quantile function equals
$\underline\theta$ on $[0,1-\pi)$ and $\bar\theta$ on $[1-\pi,1]$.
Equation \eqref{peace1} then gives
\begin{align*}
q[\nu]
&=2\int_0^{1-\pi}(1-u)^{1+\underline\theta}\,du
+2\pi^{\underline\theta-\bar\theta}
  \int_{1-\pi}^1(1-u)^{1+\bar\theta}\,du\\
&=\bigl(1-\pi^{2+\underline\theta}\bigr)
  \frac{2}{2+\underline\theta}
  +\pi^{2+\underline\theta}\frac{2}{2+\bar\theta},
\end{align*}
which is the stated upper bound.

For the limiting claim in the text, hold $\underline\theta$ and $\mu$ fixed.
As $\bar\theta\to\infty$,
$\pi=(\mu-\underline\theta)/(\bar\theta-\underline\theta)\to0$, and hence the
last display converges to $2/(2+\underline\theta)$.  The endpoint
distributions attain the bound for every finite $\bar\theta$ and thus
approach this limit.  If $\mu>\underline\theta$, attaining the limit itself
would require zero probability above $\underline\theta$, which is
inconsistent with the fixed mean.
\end{proof}

\subsection*{One-sided information}

We now keep one party's distribution fixed and change the other's.
The symmetric concession formula no longer determines their behavior:
we must first identify which types choose the same concession level and
whether some types concede immediately.  Lemma \ref{wild-type-uniqueness}
settles this matching when wild types have positive probability.  The next
lemma takes that probability to zero, giving the selected asymmetric
equilibrium needed for both parts of Theorem \ref{onesided}.

For a quantile function $Q$ on $[0,1]$, define the limiting matching index
and rank weight by
\begin{equation}\label{matching-index}
\Lambda_Q(w)=\int_{1/2}^w\frac{ds}{Q(s)(1-s)},
\qquad
\psi_Q(w)=Q(w)(1-w).
\end{equation}
Thus $\Lambda_Q'=1/\psi_Q$.  This is the $c\uparrow1$ counterpart of
$\Lambda_{i,c}$ in Lemma \ref{wild-type-uniqueness}.  Although each index
diverges at the top, the difference that determines the rank matching can have a
finite limit.

\begin{lemma}[Selected rank matching in the wild-type limit]\label{limit-matching}
Let $Q_1,Q_2$ be quantile functions whose distributions have continuous
positive densities on the interior of their common support $[0,1]$, and
suppose that
\begin{equation}\label{kappa-condition}
\kappa=\int_{1/2}^1
\frac{Q_1(s)-Q_2(s)}{Q_1(s)Q_2(s)(1-s)}\,ds
\end{equation}
converges absolutely.  This condition holds if $Q_1=Q_2$ near rank one, or
if both endpoint densities are continuous and positive.  Then, as
$c\uparrow1$, the rank-matching function $v_c(u)$ of Lemma
\ref{wild-type-uniqueness} converges, at every party-1 rank $u$ satisfying
$\Lambda_{Q_1}(u)+\kappa>\Lambda_{Q_2}(0^+)$, to the unique solution
$v(u)$ of
\begin{equation}\label{limit-matching-eq}
\Lambda_{Q_2}(v(u))=\Lambda_{Q_1}(u)+\kappa.
\end{equation}
Party-1 ranks with
$\Lambda_{Q_1}(u)+\kappa<\Lambda_{Q_2}(0^+)$ concede at level zero.  If
instead $v(0^+)>0$, party-2 ranks below $v(0^+)$ concede at level zero.  On
the escalating range, the limiting planned-exposure functions satisfy
\begin{equation}\label{limit-strategies}
b_1'(u)=\frac{Q_2(v(u))}{1-u},
\qquad
b_2'(v)=\frac{Q_1(u(v))}{1-v},
\end{equation}
with $b_i=0$ at the lowest escalating rank.  If $Q_1=Q_2$, then
$\kappa=0$ and rank matching is diagonal: $v(u)=u$.
\end{lemma}

\begin{proof}{of Lemma}{\ref{limit-matching}}
In the notation of Lemma \ref{wild-type-uniqueness},
\[
C_c=\Lambda_{2,c}(1)-\Lambda_{1,c}(1)
=\int_{1/2}^1
\frac{Q_1(s)-Q_2(s)}{Q_1(s)Q_2(s)(1-cs)}\,ds.
\]
Since $0<(1-cs)^{-1}\leq(1-s)^{-1}$, dominated convergence and
\eqref{kappa-condition} give $C_c\to\kappa$.  For fixed $w\in(0,1)$,
$\Lambda_{i,c}(w)\to\Lambda_{Q_i}(w)$ by bounded convergence.  Fix $u$ with
$\Lambda_{Q_1}(u)+\kappa>\Lambda_{Q_2}(0^+)$.  Since $\Lambda_{Q_2}$ is
continuous, strictly increasing, and unbounded above,
\eqref{limit-matching-eq} has a unique solution $v(u)\in(0,1)$.

Choose $v_-<v(u)<v_+$ in $(0,1)$.  Strict monotonicity gives
\[
\Lambda_{Q_2}(v_-)<\Lambda_{Q_1}(u)+\kappa
<\Lambda_{Q_2}(v_+).
\]
The same inequalities hold with $\Lambda_{i,c}$ and $C_c$ for $c$ close to
one.  Equation \eqref{wild-type-matching} and monotonicity then place
$v_c(u)$ in $(v_-,v_+)$.  Since the bracket is arbitrary,
$v_c(u)\to v(u)$.  Ranks for which the limiting right-hand side lies below
$\Lambda_{Q_2}(0^+)$ have no match for all sufficiently large $c$ and form
the residual immediate-concession interval; the party-2 case is symmetric.
Equation \eqref{limit-strategies} is the pointwise limit of the
planned-exposure derivatives in the proof of Lemma
\ref{wild-type-uniqueness}.  Extend those derivatives by zero over any
immediate-concession interval.  They are bounded by $(1-u)^{-1}$ on
every fixed interval $[0,u_*]$ with $u_*<1$, so dominated convergence
also gives convergence of the integrated strategies, including their
zero lower boundary values.

If $Q_1=Q_2$, then $\kappa=0$ and rank matching is diagonal.  Finally, if
$Q_1=Q_2$ near one, the integrand in \eqref{kappa-condition} vanishes
there.  If both endpoint densities are continuous and positive, then
$|Q_1(s)-Q_2(s)|\leq C(1-s)$ near one, which also gives absolute
convergence.
\end{proof}

The matching identifies when either party plans to concede.  To compare
information structures, we now translate those plans into a probability
of peace.  It is convenient to calculate the complementary probability
that disaster arrives while both parties are still escalating; types that
concede immediately contribute nothing to this probability.

\begin{lemma}[Disaster probability under selected rank matching]
\label{disaster-formula}
In the equilibrium of Lemma \ref{limit-matching}, let $u_0$ be party 1's
lowest escalating rank and $v_0=v(u_0^+)$, so that $u_0v_0=0$, and define
\[
\mathcal Z(u)=\int_{u_0}^u\frac{Q_2(v(s))}{1-s}\,ds.
\]
Then the probability of peace $q$ satisfies
\begin{equation}\label{disaster-eq}
1-q=\int_{u_0}^1\psi_{Q_2}(v(u))e^{-\mathcal Z(u)}\,du
=\int_{u_0}^1(1-v(u))Q_2(v(u))e^{-\mathcal Z(u)}\,du.
\end{equation}
In particular, at a symmetric pair $Q_1=Q_2=\nu$,
\begin{equation}\label{symmetric-disaster-formula}
1-q[\nu]=\int_0^1(1-u)\nu(u)\sigma_\nu(u)\,du.
\end{equation}
\end{lemma}

\begin{proof}{of Lemma}{\ref{disaster-formula}}
On the disaster-hazard clock \eqref{hazard-clock}, the transformed disaster threshold has density
$e^{-b}$.  Disaster occurs only if it arrives while both parties are still
escalating.  By \eqref{limit-strategies}, $\mathcal Z(u)=b_1(u)$ is the
hazard level at which rank $u$ of party 1 and rank $v(u)$ of party 2 concede,
and both parties are still present with probability $(1-u)(1-v(u))$.
The map $u\mapsto\mathcal Z(u)$ is an increasing bijection from
$[u_0,1)$ onto $[0,\infty)$: rank matching implies $v(u)\to1$, so
$Q_2(v(u))\to1$ and the integral defining $\mathcal Z$ diverges at
the top.  Changing variables from $b$ to $u$ gives
\[
1-q=\int_{u_0}^1(1-u)(1-v(u))e^{-\mathcal Z(u)}
\mathcal Z'(u)\,du
=\int_{u_0}^1(1-v(u))Q_2(v(u))e^{-\mathcal Z(u)}\,du.
\]
Types that concede at level zero never face disaster and contribute nothing.
At a symmetric pair, $v(u)=u$, $\mathcal Z=\mathcal Z_\nu$, and
$e^{-\mathcal Z}=\sigma_\nu$.
\end{proof}

\subsubsection*{Local robustness at symmetry}

We apply the matching and disaster formulas first to the small one-sided
change in part~(i).  The main difficulty is differentiability: both
strategies adjust, and an interval of types may start conceding
immediately.  Working in type coordinates allows the proof to control
these changes under the two bounds in part~(i) and then identify the sign
of the first-order effect.

\begin{proof}{of part (i) of Theorem}{\ref{onesided}}
Write $S=1-F$, $S_\eta=1-F_\eta$, and
$\beta_F(x)=\int_0^x tf(t)/S(t)\,dt$, the symmetric planned exposure.
Define
\[
D(x)=\int_0^x\delta(z)\,dz,\qquad M(x)=\int_0^xD(y)\,dy.
\]
Since $f_\eta$ is a density with the same mean as $f$ for small positive
$\eta$, the mean-preserving-spread hypothesis implies
$\int_0^1\delta(x)\,dx=\int_0^1x\delta(x)\,dx=0$.
Hence $F_\eta=F+\eta D$ and $D(0)=D(1)=M(0)=M(1)=0$.
Retain the constants $K,\alpha$ from the theorem.  Nontriviality of the
spread implies $K>0$.
Fix $0<\eta_0<1/(2K)$ and consider
$|\eta|\leq\eta_0$.  Constants below may depend on this interval and on
$f,\delta$.  A dot denotes differentiation with respect to $\eta$,
and a prime denotes differentiation with respect to the displayed scalar
argument.  The notation $h_\eta\asymp h$ means that each function is bounded above
by a fixed positive multiple of the other, uniformly in $\eta$ and the
argument.

\medskip\noindent
\textbf{1. Feasibility and selected matching in type coordinates.}
The mass constraint gives
\begin{equation}\label{os-domination}
|D(x)|\leq K\min\{F(x),S(x)\},\qquad
f_\eta\asymp f,\qquad S_\eta:=S-\eta D\asymp S.
\end{equation}
In particular, the distributions are feasible for both signs of $\eta$.
Define, on $(0,1)$,
\begin{equation}\label{os-type-indices}
L_F'(x)=\frac{f(x)}{xS(x)},\qquad
R_\eta(x)=\log\frac{S(x)}{S_\eta(x)},\qquad
J_\eta(x)=\frac{R_\eta(x)}x-
\int_x^1\frac{R_\eta(t)}{t^2}\,dt.
\end{equation}
Domination gives $|\eta D/S|\leq\eta_0K<1/2$, so
$|R_\eta(x)|\leq2K|\eta|$ and $R_\eta(x)=O(|\eta|x^{\alpha+1})$
near zero.  The definition of $J_\eta$ in \eqref{os-type-indices} therefore
implies $|xJ_\eta(x)|\leq C|\eta|$.  Near zero, $J_\eta$ is
$O(|\eta|)$ if $\alpha>0$, but only
$O(|\eta|(1+|\log x|))$ if $\alpha=0$; a uniform bound on $J_\eta$
over the whole support is not needed.

The additive constant in $L_F$ is arbitrary when $\alpha=0$; when
$\alpha>0$, set $L_F(0)=0$.  Direct differentiation gives
\begin{equation}\label{os-index-derivative}
(L_F+J_\eta)'(x)=\frac{f_\eta(x)}{xS_\eta(x)}>0.
\end{equation}
The selected type of party 2 matched to type $x$ of party 1 is
\begin{equation}\label{os-type-matching}
y_\eta(x)=L_F^{-1}\bigl(L_F(x)+J_\eta(x)\bigr),
\end{equation}
where, if $\alpha>0$, the inverse is extended by zero on nonpositive
arguments.

To verify both applicability and normalization in Lemma
\ref{limit-matching}, write $\nu_\eta=F_\eta^{-1}$.  For any fixed upper-rank interval
$[u_*,1)$ with $u_*>0$, its lower type stays bounded away from zero and
\[
\int_{u_*}^1
\frac{|\partial_\eta(1/\nu_\eta(u))|}{1-u}\,du
=\int_{\nu_\eta(u_*)}^1
\frac{|D(x)|}{x^2S_\eta(x)}\,dx\leq C.
\]
Here $\partial_\eta\nu_\eta(u)=-D(\nu_\eta(u))/f_\eta(\nu_\eta(u))$
follows by differentiating $F_\eta(\nu_\eta(u))=u$.
Integration in $\eta$ proves the lemma's absolute tail-integrability
condition.  At equal ranks let $x=\nu_\eta(u)$ and $z=\nu_0(u)$.
Then $S(z)=S_\eta(x)$, so, as $u\uparrow1$,
\[
L_F(x)-L_F(z)=-R_\eta(x)+o(1),\qquad
J_\eta(x)=R_\eta(x)+o(1).
\]
Indeed, in the first relation replace $1/t$ by $1$ in
$\int_z^x f(t)/(tS(t))\,dt$; the error is bounded by
\[
\sup_{t\text{ between }x,z}|1/t-1|\,|R_\eta(x)|.
\]
The second follows directly from boundedness of $R_\eta$.
Thus the indices in \eqref{os-type-matching} agree at the top at equal ranks,
which is precisely the normalization selected by the common wild type.

Set $a_\eta=f_\eta/S_\eta$, the type-distribution hazard rate, and
$w_F(x)=xS(x)$, the weight appearing in the disaster probability.  The selected exposure and
disaster probability from Lemma \ref{disaster-formula} are therefore
\begin{equation}\label{os-disaster}
\begin{aligned}
b_\eta(x)&=\int_0^x y_\eta(t)a_\eta(t)\,dt,\\
I(\eta):=1-q(\eta)&=\int_0^1
w_F(y_\eta(x))f_\eta(x)e^{-b_\eta(x)}\,dx.
\end{aligned}
\end{equation}
This formula includes either possible immediate-concession mass: the
integrand vanishes on party 1's immediate-concession interval, and party
2's mass is encoded by the lower limit of its active matched types.

\medskip\noindent
\textbf{2. Endpoint control and differentiation.}
From \eqref{os-domination},
\[
\dot R_\eta(x)=\frac{D(x)}{S_\eta(x)},\qquad
\sup_{\eta,x}|x\dot J_\eta(x)|\leq C.
\]
Near zero, $\dot R_\eta(x)=O(x^{\alpha+1})$, and hence, uniformly in
$\eta$,
\begin{equation}\label{os-lower-bounds}
\dot J_\eta(x)=
\begin{cases}
O(1+|\log x|),&\alpha=0,\\
O(1),&\alpha>0.
\end{cases}
\qquad x\dot J_\eta(x)\longrightarrow0.
\end{equation}
There is no restriction on the rate at which $f$ may vanish at one.

Call the interval on which the matched types choose positive concession
levels the \emph{active matching interval}.  On this interval, differentiation of \eqref{os-type-matching}
gives
\begin{equation}\label{os-matching-derivatives}
\partial_x y_\eta(x)=\frac{f_\eta(x)}{xS_\eta(x)}
\frac{y_\eta(x)S(y_\eta(x))}{f(y_\eta(x))},\qquad
\dot y_\eta(x)=\frac{y_\eta(x)S(y_\eta(x))}{f(y_\eta(x))}
\dot J_\eta(x).
\end{equation}
The reciprocal $1/f(y)$ need not be bounded.  Instead, let $x_\eta(y)$ be
the inverse of $y_\eta$ on the active matching interval.  For $y$ below
the smallest active party-2 type, define $x_\eta(y)=0$.  Changing
variables using the first identity in \eqref{os-matching-derivatives}
cancels that inverse density.  In particular,
\begin{equation}\label{os-hazard-derivative}
\begin{aligned}
\dot b_\eta(x)&=
\int_0^{y_\eta(x)}x_\eta(y)\dot J_\eta(x_\eta(y))\,dy
+\int_0^x y_\eta(t)\dot a_\eta(t)\,dt,\\
|\dot b_\eta(x)|&\leq C(1+b_\eta(x)),
\end{aligned}
\end{equation}
where the first integrand is defined to be zero when $x_\eta(y)=0$.
Here we used
\[
\frac{\dot a_\eta}{a_\eta}
=\frac{\delta}{f_\eta}+\frac{D}{S_\eta},
\qquad |\dot a_\eta|\leq Ca_\eta.
\]

The first integral in the formula for $\dot b_\eta$ is bounded by $C$
because $|x\dot J_\eta(x)|\leq C$ and the range of integration has length
at most one.  For the second integral, retain its matching weight:
\[
\left|\int_0^x y_\eta(t)\dot a_\eta(t)\,dt\right|
\leq C\int_0^x y_\eta(t)a_\eta(t)\,dt=Cb_\eta(x).
\]
Thus no estimate for the unweighted integral $\int_0^x a_\eta(t)\,dt$
is required.

For completeness, moving immediate-concession boundaries do not prevent
absolute continuity in $\eta$.  If $\alpha=0$, $L_F(0+)=-\infty$ and there
is no boundary crossing.  If $\alpha>0$, $L_F(x)\asymp x^\alpha$ near zero
and $(L_F^{-1})'(s)=O(s^{1/\alpha-1})$.  For each fixed $x>0$, the argument
$L_F(x)+J_\eta(x)$ is real analytic in $\eta$ on $|\eta|<1/K$ (expand
$-\log(1-\eta D/S)$ in its convergent power series), and is
positive at zero.  Its zeros on a compact parameter interval are finite
in number and of finite order.  Near such a zero $\eta_*$ of order $m\geq1$, the derivative
of the zero-extended inverse is $O(|\eta-\eta_*|^{m/\alpha-1})$, which is
integrable.  Thus $y_\eta(x)$ is absolutely continuous in $\eta$.
The change of variables above also gives
$\int_0^1|\dot y_\eta(x)|a_\eta(x)\,dx\leq C$.
Together with the bound on $\dot a_\eta$, this justifies
\eqref{os-hazard-derivative} by Fubini on every fixed type interval
$[0,x]$ with $x<1$.

Differentiating \eqref{os-disaster} now yields, initially for almost every
$\eta$, $I'(\eta)=T_\eta+U_\eta-V_\eta$, where
\begin{align*}
T_\eta&=\int_0^1 e^{-b_\eta(x_\eta(y))}
x_\eta(y)S_\eta(x_\eta(y))w_F'(y)
\dot J_\eta(x_\eta(y))\,dy,\\
U_\eta&=\int_0^1 e^{-b_\eta(x)}w_F(y_\eta(x))\delta(x)\,dx,\\
V_\eta&=\int_0^1 e^{-b_\eta(x)}w_F(y_\eta(x))
f_\eta(x)\dot b_\eta(x)\,dx.
\end{align*}
Again the first integrand is zero when $x_\eta(y)=0$.
Their absolute integrands are bounded, respectively, by
\[
C|w_F'(y)|,\qquad Kf(x),\qquad
Cf(x)(1+b_\eta(x))e^{-b_\eta(x)}\leq C'f(x).
\]
All these bounds are integrable, since $w_F'=S-xf$ is continuous.
Applied also to the absolute derivatives before the change of variables,
they prove absolute continuity of $I$ by Fubini.

For each interior type, $y_\eta(x)\to x$ and $x_\eta(y)\to y$.
Moreover, $S_\eta/S\to1$ uniformly and
$\dot R_\eta=D/S_\eta\to D/S$ uniformly.  The integral formula for
$\dot J_\eta$ then gives convergence to $\dot J_0$ uniformly on compact
subsets of $(0,1)$.  Hence
$\dot J_\eta(x_\eta(y))\to\dot J_0(y)$ for each interior $y$.
On every fixed interval $[0,x]$ with $x<1$, the exposure integrands are
dominated by $Cf/S$, so $b_\eta\to\beta_F$ uniformly there.  In particular,
$b_\eta(x_\eta(y))\to\beta_F(y)$ for each interior $y$.  The first integral in
\eqref{os-hazard-derivative} has a uniformly bounded integrand; the second
is dominated on $[0,x]$ by a constant times $f/S$.  Thus
$\dot b_\eta(x)\to\dot b_0(x)$.  Dominated convergence in $T_\eta$,
$U_\eta$, and $V_\eta$ proves that their sum is continuous at zero.
To see why this gives a two-sided derivative, absolute continuity and
the almost-everywhere derivative formula imply
\[
\frac{I(\eta)-I(0)}{\eta}
=\frac1\eta\int_0^\eta(T_s+U_s-V_s)\,ds
\longrightarrow T_0+U_0-V_0.
\]
The convergence follows from continuity of the integrand at zero and
holds for either sign of $\eta$.

The size of any concession mass can also be read off directly.  If
$\alpha>0$, $J_\eta(0)=O(|\eta|)$ and
$(L_F+J_\eta)'(x)\asymp x^{\alpha-1}$ near zero.  Either cutoff type is
$O(|\eta|^{1/\alpha})$, and its probability mass is
$O(|\eta|^{1+1/\alpha})=o(|\eta|)$.  If $\alpha=0$, there is no such
mass.  Thus immediate concessions do not contribute a separate first-order
boundary term.

\medskip\noindent
\textbf{3. The derivative and its sign.}
To evaluate the first variation, abbreviate
\[
\begin{aligned}
d(x)&=\frac{D(x)}{S(x)},&
H_D(x)&=\int_x^1\frac{d(t)}{t^2}\,dt,\\
B_D(x)&=\int_0^x tH_D(t)\,dt,&
\Theta_F(x)&=\int_x^1 e^{-\beta_F(t)}S(t)f(t)\,dt.
\end{aligned}
\]
At symmetry,
\[
\dot J_0(x)=\frac{d(x)}x-H_D(x),\qquad
\dot b_0(x)=xd(x)-B_D(x).
\]
For the second identity, use $\dot a_0=d'$ in
$\dot b_0=\int_0^x[t\dot J_0(t)+td'(t)]\,dt$ and integrate by parts.
Substituting these identities in the first variation gives
\[
\begin{aligned}
I'(0)=\int_0^1e^{-\beta_F(x)}\bigl[&xS(x)\delta(x)
+xS(x)(S(x)-xf(x))\dot J_0(x)\\
&-xS(x)f(x)\dot b_0(x)\bigr]dx.
\end{aligned}
\]
Integration by parts in the term containing $\delta=D'$ cancels all
terms containing $d$ directly.  The remaining expression is
\[
I'(0)=\int_0^1e^{-\beta_F(x)}
\left[-xS(x)(S(x)-xf(x))H_D(x)+xS(x)f(x)B_D(x)\right]dx.
\]
Since $(S^2e^{-\beta_F})'=-(2+x)Sfe^{-\beta_F}$,
\[
\int_x^1 tS(t)f(t)e^{-\beta_F(t)}\,dt
=S(x)^2e^{-\beta_F(x)}-2\Theta_F(x).
\]
Exchanging the order in the term containing $B_D$ therefore gives
\[
I'(0)=\int_0^1 xH_D(x)
\left[xS(x)f(x)e^{-\beta_F(x)}-2\Theta_F(x)\right]dx
=-\int_0^1 d(x)\Theta_F(x)\,dx.
\]
The last equality integrates $-H_D(x)(x^2\Theta_F(x))'$ by parts, using
$H_D'=-d/x^2$.  These manipulations are absolutely convergent: $d$ is
bounded and is $O(x^{\alpha+1})$ at zero, while $H_D$ is bounded there if
$\alpha>0$ and is $O(1+|\log x|)$ if $\alpha=0$.

Consequently,
\begin{equation}\label{first-variation-final}
q'(0)=\int_0^1\frac{D(x)\Theta_F(x)}{S(x)}\,dx.
\end{equation}
Define $A_F(x)=\Theta_F(x)/S(x)$.  Then $0\leq A_F\leq S/2$ and
\[
-A_F'(x)=\frac{f(x)}{S(x)^2}
\left[S(x)^2e^{-\beta_F(x)}-\Theta_F(x)\right]
=: \mathcal K_F(x),
\]
The bracket has derivative
\[
-(1+x)f(x)S(x)e^{-\beta_F(x)}
\]
and limit zero at one.  It therefore equals
$\int_x^1(1+t)f(t)S(t)e^{-\beta_F(t)}\,dt$, giving the integral
representation of $\mathcal K_F$.
Since $\Theta_F(x)\leq e^{-\beta_F(x)}S(x)^2/2$,
\[
\tfrac12f(x)e^{-\beta_F(x)}\leq\mathcal K_F(x)
\leq f(x)e^{-\beta_F(x)}\leq f(x)
\qquad(0<x<1).
\]
The moment constraints imply $M(0)=M(1)=0$, so integration by parts in
$D=M'$ gives
\begin{equation}\label{onesided-kernel}
q'(0)=\int_0^1\mathcal K_F(x)M(x)\,dx,\qquad
\mathcal K_F(x)=\frac{f(x)}{S(x)^2}
\int_x^1(1+t)f(t)S(t)e^{-\beta_F(t)}\,dt.
\end{equation}
The mean-preserving-spread hypothesis gives $M\geq0$ and $M\not\equiv0$.
Consequently,
\begin{equation}\label{onesided-derivative-bound}
q'(0)\geq\frac12\int_0^1f(x)e^{-\beta_F(x)}M(x)\,dx>0.
\end{equation}
The expansion $q(\eta)=q(0)+\eta q'(0)+o(\eta)$ proves part~(i) and
also the reverse-direction conclusion in Remark \ref{rem:onesided-reverse}.

\end{proof}

The local calculation also gives the half-effect identity.  We compare
its derivative with the derivative obtained when both parties'
distributions change along the same path.

\begin{proof}{of Corollary}{\ref{onesided-half}}
Retain $D$, $S_\eta$, $\beta_F$, and
$\Theta_F(x)=\int_x^1 e^{-\beta_F(t)}S(t)f(t)\,dt$ from the local proof,
and write $\nu_\eta=F_\eta^{-1}$ for the quantile function.
The symmetric comparison needs no uniform pointwise bound on the
quantile derivative.  The following weighted integral suffices:
\[
\int_0^1\frac{|\partial_\eta\nu_\eta(u)|}{1-u}\,du
=\int_0^1\frac{|D(x)|}{S_\eta(x)}\,dx\leq C.
\]
Near zero, the lower-endpoint bounds also give
$|\partial_\eta\nu_\eta(u)|\leq Cu^{1/(\alpha+1)}$ uniformly in $\eta$;
on compact interior rank intervals the quantile derivatives are bounded
and continuous in $\eta$.  Differentiation under the inner integral is
therefore valid for each $u<1$.  The weighted bound above then bounds
that derivative uniformly in $u$, justifying differentiation of the
symmetric peace formula
\[
q[F_\eta,F_\eta]=2\int_0^1(1-u)
\exp\left\{-\int_0^u\frac{\nu_\eta(v)}{1-v}\,dv\right\}du.
\]
Using $\partial_\eta\nu_\eta|_0=-D(\nu_0)/f(\nu_0)$ and exchanging
integrals gives
\[
\left.\frac{d}{d\eta}q[F_\eta,F_\eta]\right|_0
=2\int_0^1\frac{D(x)\Theta_F(x)}{S(x)}\,dx=2q'(0).
\]
The last equality is \eqref{first-variation-final}, proving
\eqref{onesided-half-identity}.
\end{proof}

\subsubsection*{Reversal away from symmetry}

The local result does not determine the sign of an information change at
an asymmetric starting point.  To prove part~(ii), we first construct a
one-sided information improvement that lowers peace between two
endpoints.  We then use the garbling path from the main text to locate an
asymmetric point at which arbitrarily small improvements lower peace.

\begin{proof}{of part (ii) of Theorem}{\ref{onesided}}
Let party 2's posterior value be uniform on $[0,1]$.  For an integer $n>1$,
let party 1's posterior-value quantile be
\[
Q_n(u)=\frac{u^{1/n}}{u^{1/n}+(1-u)^{1/n}}.
\]
It is the quantile function of
\[
F_n(x)=\frac{x^n}{x^n+(1-x)^n}.
\]
The symmetry $Q_n(1-u)=1-Q_n(u)$ gives mean $1/2$, and
$Q_n(u)>u$ below $1/2$ and $Q_n(u)<u$ above $1/2$.  Hence $F_n$ is a
mean-preserving contraction of the uniform distribution.  By the martingale
characterization of posterior means, party 1 can first observe a signal whose
posterior-value distribution is $F_n$ and then learn its uniform benefit
fully.  This is a Blackwell improvement.

Let $q_n^0$ denote the selected equilibrium probability of peace before full
revelation.  To bound its complement, define
\[
\Lambda_{Q_n}(u)=\int_{1/2}^u\frac{ds}{Q_n(s)(1-s)}.
\]
Before applying Lemma \ref{limit-matching}, verify its tail condition.
The identity $Q_n(s)^{-1}=1+((1-s)/s)^{1/n}$ gives
\[
\left|\frac1s-\frac1{Q_n(s)}\right|\frac1{1-s}
\leq\frac1s+\frac{(1-s)^{1/n-1}}{s^{1/n}}.
\]
The right-hand side is integrable on $[1/2,1]$, since $1/n-1>-1$.
Thus \eqref{kappa-condition} holds even though $f_n$ vanishes at both
endpoints.  By Lemma \ref{limit-matching}, with $Q_1=Q_n$ and $Q_2(v)=v$, the selected
rank-matching function satisfies
\[
\log\frac{v_n(u)}{1-v_n(u)}=\Lambda_{Q_n}(u)+\kappa_n,
\]
where
\[
\begin{aligned}
\kappa_n
&=\int_{1/2}^1
\left[\frac1s-\frac1{Q_n(s)}\right]\frac{ds}{1-s}\\
&=\log2-\int_{1/2}^1
(1-s)^{1/n-1}s^{-1/n}\,ds
=\log2-n\int_0^1\frac{dr}{1+r^n}
\leq\log2-\frac n2.
\end{aligned}
\]
Here we used
$(1-Q_n(s))/Q_n(s)=((1-s)/s)^{1/n}$ and the substitution
$r=((1-s)/s)^{1/n}$.  Thus
\[
v_n(u)=\frac{\exp[\Lambda_{Q_n}(u)+\kappa_n]}
{1+\exp[\Lambda_{Q_n}(u)+\kappa_n]}.
\]
Because $\Lambda_{Q_n}(0)$ is finite for $n>1$, $v_n(0^+)>0$: party 2 has an
immediate-concession atom consisting of ranks below $v_n(0^+)$.  The
rank-matching function and the disaster formula below already account for this atom.
For $u\leq1/2$, $\Lambda_{Q_n}(u)\leq0$.  For
$u\in[1/2,4/5]$, $Q_n(u)\geq1/2$, and therefore
\[
\Lambda_{Q_n}(u)\leq2\log\frac52.
\]
Consequently, for every $u\leq4/5$,
\[
v_n(u)\leq e^{\Lambda_{Q_n}(u)+\kappa_n}
\leq\frac{25}{2}e^{-n/2}.
\]

Let
\[
\mathcal Z_n(u)=\int_0^u\frac{v_n(s)}{1-s}\,ds.
\]
Lemma \ref{disaster-formula}, specialized to $Q_2(v)=v$, gives the
probability of disaster as
\[
1-q_n^0=\int_0^1e^{-\mathcal Z_n(u)}v_n(u)[1-v_n(u)]\,du.
\]
Using $e^{-\mathcal Z_n}\leq1$, the preceding bound on $[0,4/5]$, and
$v_n(1-v_n)\leq1/4$ on the remaining interval gives
\[
1-q_n^0
\leq10e^{-n/2}+\frac1{20}.
\]
For $n=20$ this implies
\[
q_{20}^0\geq\frac{19}{20}-10e^{-10}>0.94.
\]

After full revelation both parties are uniform.  Let $q^1$ denote the
resulting probability of peace.  Direct evaluation of the symmetric peace
functional gives
\[
q^1=2\int_0^1e^u(1-u)^2\,du=4e-10<0.88.
\]
Thus $q^1<q_{20}^0$.  To obtain an adverse \emph{local} effect at an
asymmetric pair, let $q_t$ be the selected probability along
\eqref{onesided-family}; its endpoints satisfy $q_1>q_0$.
We show that $q_t$ is continuous on $[0,1]$ and differentiable on
$(0,1)$, and then apply the mean value theorem.

Write $Q_t=F_t^{-1}$ and $L_t(u)=\log[v_t(u)/(1-v_t(u))]$.
On any compact parameter interval contained in $[0,1)$, $f_t$ is bounded
above and away from zero.  Implicit differentiation of $F_t(Q_t(u))=u$
therefore gives $Q_t(u)\geq c_0u$ and
$|\partial_tQ_t(u)|\leq C\min\{u,1-u\}$.  Indeed,
\[
\partial_tQ_t(u)
=-\frac{F_{20}(Q_t(u))-Q_t(u)}{f_t(Q_t(u))}.
\]
Here $F_{20}$ denotes the distribution at the endpoint of the path:
\[
F_{20}(x)=\frac{x^{20}}{x^{20}+(1-x)^{20}}.
\]
The numerator is bounded by $C\min\{Q_t(u),1-Q_t(u)\}$ because it
vanishes at both endpoints and has bounded derivative; the density bounds
make $Q_t(u)$ and $1-Q_t(u)$ comparable to $u$ and $1-u$, respectively.  Differentiating the matching formula yields
\[
\partial_tL_t(u)
=\int_u^1\frac{\partial_tQ_t(s)}{Q_t(s)^2(1-s)}\,ds.
\]
This derivative is $O(1+|\log u|)$ near zero and $O(1-u)$ near one.
Since $\partial_tv_t=v_t(1-v_t)\partial_tL_t$, these bounds justify
differentiation of $\mathcal Z_t(u)=\int_0^uv_t(s)/(1-s)\,ds$ and
of the disaster integral \eqref{disaster-eq} by dominated convergence.
They also give continuity at $t=0$.
At $t=1$, use $Q_{20}(s)\leq Q_t(s)\leq s$ for $s\geq1/2$:
the absolute integrand defining the matching constant is bounded by the
integrable endpoint integrand already checked above.  Thus the matching
constants, matched ranks, and disaster integrals converge as $t\uparrow1$,
again by dominated convergence.

Choose $t_b<1$ with $q_{t_b}>q_0$.  The mean value theorem gives some
$t_*\in(0,t_b)$ with
$\left.dq_t/dt\right|_{t=t_*}>0$.  Consequently
$q_{t_*-h}<q_{t_*}$ for every sufficiently small $h>0$.
Decreasing $t$ is a Blackwell improvement along this family, while
$t_*>0$ means that the starting pair is asymmetric.  This proves part~(ii)
without asserting that a sequence of improvements starting from symmetry
ever brings peace below its initial symmetric level.
\end{proof}

\end{document}